\documentclass[11pt]{amsart}
\usepackage{amsmath,amssymb,amsthm,amscd,verbatim}
\usepackage{graphicx}
\usepackage{epsfig}
\usepackage{hyperref}
\usepackage[noabbrev,capitalize]{cleveref}
\usepackage{thmtools}
\usepackage{thm-restate}
\usepackage{optidef}

\usepackage{tikz}
\usepackage{pgfplots}

    \usepackage[breakable]{tcolorbox}
    \usepackage{parskip} % Stop auto-indenting (to mimic markdown behaviour)
    
    \usepackage{graphicx}
    \usepackage[Export]{adjustbox} % Used to constrain images to a maximum size
    \adjustboxset{max size={0.9\linewidth}{0.9\paperheight}}
    \usepackage{float}
    \usepackage{xcolor} % Allow colors to be defined
    \usepackage{enumerate} % Needed for markdown enumerations to work
    \usepackage{geometry} % Used to adjust the document margins
    \usepackage{amsmath} % Equations
    \usepackage{amssymb} % Equations
    \usepackage{textcomp} % defines textquotesingle
    \AtBeginDocument{%
        \def\PYZsq{\textquotesingle}% Upright quotes in Pygmentized code
    }
    \usepackage{upquote} % Upright quotes for verbatim code
    \usepackage{eurosym} % defines \euro
    \usepackage[mathletters]{ucs} % Extended unicode (utf-8) support
    \usepackage{fancyvrb} % verbatim replacement that allows latex
    \usepackage{grffile} % extends the file name processing of package graphics 
    \makeatletter % fix for grffile with XeLaTeX
    \def\Gread@@xetex#1{%
      \IfFileExists{"\Gin@base".bb}%
      {\Gread@eps{\Gin@base.bb}}%
      {\Gread@@xetex@aux#1}%
    }
    \makeatother

    \usepackage{hyperref}
    \usepackage{longtable} % longtable support required by pandoc >1.10
    \usepackage{booktabs}  % table support for pandoc > 1.12.2
    \usepackage[inline]{enumitem} % IRkernel/repr support (it uses the enumerate* environment)
    \usepackage[normalem]{ulem} % ulem is needed to support strikethroughs (\sout)
    \usepackage{mathrsfs}

    \definecolor{urlcolor}{rgb}{0,.145,.698}
    \definecolor{linkcolor}{rgb}{.71,0.21,0.01}
    \definecolor{citecolor}{rgb}{.12,.54,.11}

    \definecolor{ansi-black}{HTML}{3E424D}
    \definecolor{ansi-black-intense}{HTML}{282C36}
    \definecolor{ansi-red}{HTML}{E75C58}
    \definecolor{ansi-red-intense}{HTML}{B22B31}
    \definecolor{ansi-green}{HTML}{00A250}
    \definecolor{ansi-green-intense}{HTML}{007427}
    \definecolor{ansi-yellow}{HTML}{DDB62B}
    \definecolor{ansi-yellow-intense}{HTML}{B27D12}
    \definecolor{ansi-blue}{HTML}{208FFB}
    \definecolor{ansi-blue-intense}{HTML}{0065CA}
    \definecolor{ansi-magenta}{HTML}{D160C4}
    \definecolor{ansi-magenta-intense}{HTML}{A03196}
    \definecolor{ansi-cyan}{HTML}{60C6C8}
    \definecolor{ansi-cyan-intense}{HTML}{258F8F}
    \definecolor{ansi-white}{HTML}{C5C1B4}
    \definecolor{ansi-white-intense}{HTML}{A1A6B2}
    \definecolor{ansi-default-inverse-fg}{HTML}{FFFFFF}
    \definecolor{ansi-default-inverse-bg}{HTML}{000000}

    \DefineVerbatimEnvironment{Highlighting}{Verbatim}{commandchars=\\\{\}}
    \let\Oldtex\TeX
    \let\Oldlatex\LaTeX
    \renewcommand{\TeX}{\textrm{\Oldtex}}
    \renewcommand{\LaTeX}{\textrm{\Oldlatex}}
\makeatletter
\def\PY@reset{\let\PY@it=\relax \let\PY@bf=\relax%
    \let\PY@ul=\relax \let\PY@tc=\relax%
    \let\PY@bc=\relax \let\PY@ff=\relax}
\def\PY@tok#1{\csname PY@tok@#1\endcsname}
\def\PY@toks#1+{\ifx\relax#1\empty\else%
    \PY@tok{#1}\expandafter\PY@toks\fi}
\def\PY@do#1{\PY@bc{\PY@tc{\PY@ul{%
    \PY@it{\PY@bf{\PY@ff{#1}}}}}}}
\def\PY#1#2{\PY@reset\PY@toks#1+\relax+\PY@do{#2}}

\expandafter\def\csname PY@tok@w\endcsname{\def\PY@tc##1{\textcolor[rgb]{0.73,0.73,0.73}{##1}}}
\expandafter\def\csname PY@tok@c\endcsname{\let\PY@it=\textit\def\PY@tc##1{\textcolor[rgb]{0.25,0.50,0.50}{##1}}}
\expandafter\def\csname PY@tok@cp\endcsname{\def\PY@tc##1{\textcolor[rgb]{0.74,0.48,0.00}{##1}}}
\expandafter\def\csname PY@tok@k\endcsname{\let\PY@bf=\textbf\def\PY@tc##1{\textcolor[rgb]{0.00,0.50,0.00}{##1}}}
\expandafter\def\csname PY@tok@kp\endcsname{\def\PY@tc##1{\textcolor[rgb]{0.00,0.50,0.00}{##1}}}
\expandafter\def\csname PY@tok@kt\endcsname{\def\PY@tc##1{\textcolor[rgb]{0.69,0.00,0.25}{##1}}}
\expandafter\def\csname PY@tok@o\endcsname{\def\PY@tc##1{\textcolor[rgb]{0.40,0.40,0.40}{##1}}}
\expandafter\def\csname PY@tok@ow\endcsname{\let\PY@bf=\textbf\def\PY@tc##1{\textcolor[rgb]{0.67,0.13,1.00}{##1}}}
\expandafter\def\csname PY@tok@nb\endcsname{\def\PY@tc##1{\textcolor[rgb]{0.00,0.50,0.00}{##1}}}
\expandafter\def\csname PY@tok@nf\endcsname{\def\PY@tc##1{\textcolor[rgb]{0.00,0.00,1.00}{##1}}}
\expandafter\def\csname PY@tok@nc\endcsname{\let\PY@bf=\textbf\def\PY@tc##1{\textcolor[rgb]{0.00,0.00,1.00}{##1}}}
\expandafter\def\csname PY@tok@nn\endcsname{\let\PY@bf=\textbf\def\PY@tc##1{\textcolor[rgb]{0.00,0.00,1.00}{##1}}}
\expandafter\def\csname PY@tok@ne\endcsname{\let\PY@bf=\textbf\def\PY@tc##1{\textcolor[rgb]{0.82,0.25,0.23}{##1}}}
\expandafter\def\csname PY@tok@nv\endcsname{\def\PY@tc##1{\textcolor[rgb]{0.10,0.09,0.49}{##1}}}
\expandafter\def\csname PY@tok@no\endcsname{\def\PY@tc##1{\textcolor[rgb]{0.53,0.00,0.00}{##1}}}
\expandafter\def\csname PY@tok@nl\endcsname{\def\PY@tc##1{\textcolor[rgb]{0.63,0.63,0.00}{##1}}}
\expandafter\def\csname PY@tok@ni\endcsname{\let\PY@bf=\textbf\def\PY@tc##1{\textcolor[rgb]{0.60,0.60,0.60}{##1}}}
\expandafter\def\csname PY@tok@na\endcsname{\def\PY@tc##1{\textcolor[rgb]{0.49,0.56,0.16}{##1}}}
\expandafter\def\csname PY@tok@nt\endcsname{\let\PY@bf=\textbf\def\PY@tc##1{\textcolor[rgb]{0.00,0.50,0.00}{##1}}}
\expandafter\def\csname PY@tok@nd\endcsname{\def\PY@tc##1{\textcolor[rgb]{0.67,0.13,1.00}{##1}}}
\expandafter\def\csname PY@tok@s\endcsname{\def\PY@tc##1{\textcolor[rgb]{0.73,0.13,0.13}{##1}}}
\expandafter\def\csname PY@tok@sd\endcsname{\let\PY@it=\textit\def\PY@tc##1{\textcolor[rgb]{0.73,0.13,0.13}{##1}}}
\expandafter\def\csname PY@tok@si\endcsname{\let\PY@bf=\textbf\def\PY@tc##1{\textcolor[rgb]{0.73,0.40,0.53}{##1}}}
\expandafter\def\csname PY@tok@se\endcsname{\let\PY@bf=\textbf\def\PY@tc##1{\textcolor[rgb]{0.73,0.40,0.13}{##1}}}
\expandafter\def\csname PY@tok@sr\endcsname{\def\PY@tc##1{\textcolor[rgb]{0.73,0.40,0.53}{##1}}}
\expandafter\def\csname PY@tok@ss\endcsname{\def\PY@tc##1{\textcolor[rgb]{0.10,0.09,0.49}{##1}}}
\expandafter\def\csname PY@tok@sx\endcsname{\def\PY@tc##1{\textcolor[rgb]{0.00,0.50,0.00}{##1}}}
\expandafter\def\csname PY@tok@m\endcsname{\def\PY@tc##1{\textcolor[rgb]{0.40,0.40,0.40}{##1}}}
\expandafter\def\csname PY@tok@gh\endcsname{\let\PY@bf=\textbf\def\PY@tc##1{\textcolor[rgb]{0.00,0.00,0.50}{##1}}}
\expandafter\def\csname PY@tok@gu\endcsname{\let\PY@bf=\textbf\def\PY@tc##1{\textcolor[rgb]{0.50,0.00,0.50}{##1}}}
\expandafter\def\csname PY@tok@gd\endcsname{\def\PY@tc##1{\textcolor[rgb]{0.63,0.00,0.00}{##1}}}
\expandafter\def\csname PY@tok@gi\endcsname{\def\PY@tc##1{\textcolor[rgb]{0.00,0.63,0.00}{##1}}}
\expandafter\def\csname PY@tok@gr\endcsname{\def\PY@tc##1{\textcolor[rgb]{1.00,0.00,0.00}{##1}}}
\expandafter\def\csname PY@tok@ge\endcsname{\let\PY@it=\textit}
\expandafter\def\csname PY@tok@gs\endcsname{\let\PY@bf=\textbf}
\expandafter\def\csname PY@tok@gp\endcsname{\let\PY@bf=\textbf\def\PY@tc##1{\textcolor[rgb]{0.00,0.00,0.50}{##1}}}
\expandafter\def\csname PY@tok@go\endcsname{\def\PY@tc##1{\textcolor[rgb]{0.53,0.53,0.53}{##1}}}
\expandafter\def\csname PY@tok@gt\endcsname{\def\PY@tc##1{\textcolor[rgb]{0.00,0.27,0.87}{##1}}}
\expandafter\def\csname PY@tok@err\endcsname{\def\PY@bc##1{\setlength{\fboxsep}{0pt}\fcolorbox[rgb]{1.00,0.00,0.00}{1,1,1}{\strut ##1}}}
\expandafter\def\csname PY@tok@kc\endcsname{\let\PY@bf=\textbf\def\PY@tc##1{\textcolor[rgb]{0.00,0.50,0.00}{##1}}}
\expandafter\def\csname PY@tok@kd\endcsname{\let\PY@bf=\textbf\def\PY@tc##1{\textcolor[rgb]{0.00,0.50,0.00}{##1}}}
\expandafter\def\csname PY@tok@kn\endcsname{\let\PY@bf=\textbf\def\PY@tc##1{\textcolor[rgb]{0.00,0.50,0.00}{##1}}}
\expandafter\def\csname PY@tok@kr\endcsname{\let\PY@bf=\textbf\def\PY@tc##1{\textcolor[rgb]{0.00,0.50,0.00}{##1}}}
\expandafter\def\csname PY@tok@bp\endcsname{\def\PY@tc##1{\textcolor[rgb]{0.00,0.50,0.00}{##1}}}
\expandafter\def\csname PY@tok@fm\endcsname{\def\PY@tc##1{\textcolor[rgb]{0.00,0.00,1.00}{##1}}}
\expandafter\def\csname PY@tok@vc\endcsname{\def\PY@tc##1{\textcolor[rgb]{0.10,0.09,0.49}{##1}}}
\expandafter\def\csname PY@tok@vg\endcsname{\def\PY@tc##1{\textcolor[rgb]{0.10,0.09,0.49}{##1}}}
\expandafter\def\csname PY@tok@vi\endcsname{\def\PY@tc##1{\textcolor[rgb]{0.10,0.09,0.49}{##1}}}
\expandafter\def\csname PY@tok@vm\endcsname{\def\PY@tc##1{\textcolor[rgb]{0.10,0.09,0.49}{##1}}}
\expandafter\def\csname PY@tok@sa\endcsname{\def\PY@tc##1{\textcolor[rgb]{0.73,0.13,0.13}{##1}}}
\expandafter\def\csname PY@tok@sb\endcsname{\def\PY@tc##1{\textcolor[rgb]{0.73,0.13,0.13}{##1}}}
\expandafter\def\csname PY@tok@sc\endcsname{\def\PY@tc##1{\textcolor[rgb]{0.73,0.13,0.13}{##1}}}
\expandafter\def\csname PY@tok@dl\endcsname{\def\PY@tc##1{\textcolor[rgb]{0.73,0.13,0.13}{##1}}}
\expandafter\def\csname PY@tok@s2\endcsname{\def\PY@tc##1{\textcolor[rgb]{0.73,0.13,0.13}{##1}}}
\expandafter\def\csname PY@tok@sh\endcsname{\def\PY@tc##1{\textcolor[rgb]{0.73,0.13,0.13}{##1}}}
\expandafter\def\csname PY@tok@s1\endcsname{\def\PY@tc##1{\textcolor[rgb]{0.73,0.13,0.13}{##1}}}
\expandafter\def\csname PY@tok@mb\endcsname{\def\PY@tc##1{\textcolor[rgb]{0.40,0.40,0.40}{##1}}}
\expandafter\def\csname PY@tok@mf\endcsname{\def\PY@tc##1{\textcolor[rgb]{0.40,0.40,0.40}{##1}}}
\expandafter\def\csname PY@tok@mh\endcsname{\def\PY@tc##1{\textcolor[rgb]{0.40,0.40,0.40}{##1}}}
\expandafter\def\csname PY@tok@mi\endcsname{\def\PY@tc##1{\textcolor[rgb]{0.40,0.40,0.40}{##1}}}
\expandafter\def\csname PY@tok@il\endcsname{\def\PY@tc##1{\textcolor[rgb]{0.40,0.40,0.40}{##1}}}
\expandafter\def\csname PY@tok@mo\endcsname{\def\PY@tc##1{\textcolor[rgb]{0.40,0.40,0.40}{##1}}}
\expandafter\def\csname PY@tok@ch\endcsname{\let\PY@it=\textit\def\PY@tc##1{\textcolor[rgb]{0.25,0.50,0.50}{##1}}}
\expandafter\def\csname PY@tok@cm\endcsname{\let\PY@it=\textit\def\PY@tc##1{\textcolor[rgb]{0.25,0.50,0.50}{##1}}}
\expandafter\def\csname PY@tok@cpf\endcsname{\let\PY@it=\textit\def\PY@tc##1{\textcolor[rgb]{0.25,0.50,0.50}{##1}}}
\expandafter\def\csname PY@tok@c1\endcsname{\let\PY@it=\textit\def\PY@tc##1{\textcolor[rgb]{0.25,0.50,0.50}{##1}}}
\expandafter\def\csname PY@tok@cs\endcsname{\let\PY@it=\textit\def\PY@tc##1{\textcolor[rgb]{0.25,0.50,0.50}{##1}}}

\def\PYZus{\char`\_}

\def\PYZca{\char`\^}

\def\PYZlt{\char`\<}

\def\PYZhy{\char`\-}
\def\PYZsq{\char`\'}

\makeatother

    \makeatletter
        \newbox\Wrappedcontinuationbox 
        \newbox\Wrappedvisiblespacebox 
        \newcommand*\Wrappedvisiblespace {\textcolor{red}{\textvisiblespace}} 
        \newcommand*\Wrappedcontinuationsymbol {\textcolor{red}{\llap{\tiny$\m@th\hookrightarrow$}}} 
        \newcommand*\Wrappedcontinuationindent {3ex } 
        \newcommand*\Wrappedafterbreak {\kern\Wrappedcontinuationindent\copy\Wrappedcontinuationbox} 
        \newcommand*\Wrappedbreaksatspecials {% 
            \def\PYGZus{\discretionary{\char`\_}{\Wrappedafterbreak}{\char`\_}}% 
            \def\PYGZob{\discretionary{}{\Wrappedafterbreak\char`\{}{\char`\{}}% 
            \def\PYGZcb{\discretionary{\char`\}}{\Wrappedafterbreak}{\char`\}}}% 
            \def\PYGZca{\discretionary{\char`\^}{\Wrappedafterbreak}{\char`\^}}% 
            \def\PYGZam{\discretionary{\char`\&}{\Wrappedafterbreak}{\char`\&}}% 
            \def\PYGZlt{\discretionary{}{\Wrappedafterbreak\char`\<}{\char`\<}}% 
            \def\PYGZgt{\discretionary{\char`\>}{\Wrappedafterbreak}{\char`\>}}% 
            \def\PYGZsh{\discretionary{}{\Wrappedafterbreak\char`\#}{\char`\#}}% 
            \def\PYGZpc{\discretionary{}{\Wrappedafterbreak\char`\%}{\char`\%}}% 
            \def\PYGZdl{\discretionary{}{\Wrappedafterbreak\char`\$}{\char`\$}}% 
            \def\PYGZhy{\discretionary{\char`\-}{\Wrappedafterbreak}{\char`\-}}% 
            \def\PYGZsq{\discretionary{}{\Wrappedafterbreak\textquotesingle}{\textquotesingle}}% 
            \def\PYGZdq{\discretionary{}{\Wrappedafterbreak\char`\"}{\char`\"}}% 
            \def\PYGZti{\discretionary{\char`\~}{\Wrappedafterbreak}{\char`\~}}% 
        } 
        \newcommand*\Wrappedbreaksatpunct {% 
            \lccode`\~`\.\lowercase{\def~}{\discretionary{\hbox{\char`\.}}{\Wrappedafterbreak}{\hbox{\char`\.}}}% 
            \lccode`\~`\,\lowercase{\def~}{\discretionary{\hbox{\char`\,}}{\Wrappedafterbreak}{\hbox{\char`\,}}}% 
            \lccode`\~`\;\lowercase{\def~}{\discretionary{\hbox{\char`\;}}{\Wrappedafterbreak}{\hbox{\char`\;}}}% 
            \lccode`\~`\:\lowercase{\def~}{\discretionary{\hbox{\char`\:}}{\Wrappedafterbreak}{\hbox{\char`\:}}}% 
            \lccode`\~`\?\lowercase{\def~}{\discretionary{\hbox{\char`\?}}{\Wrappedafterbreak}{\hbox{\char`\?}}}% 
            \lccode`\~`\!\lowercase{\def~}{\discretionary{\hbox{\char`\!}}{\Wrappedafterbreak}{\hbox{\char`\!}}}% 
            \lccode`\~`\/\lowercase{\def~}{\discretionary{\hbox{\char`\/}}{\Wrappedafterbreak}{\hbox{\char`\/}}}% 
            \catcode`\.\active
            \catcode`\,\active 
            \catcode`\;\active
            \catcode`\:\active
            \catcode`\?\active
            \catcode`\!\active
            \catcode`\/\active 
            \lccode`\~`\~ 	
        }
    \makeatother

    \let\OriginalVerbatim=\Verbatim
    \makeatletter
    \renewcommand{\Verbatim}[1][1]{%
        \sbox\Wrappedcontinuationbox {\Wrappedcontinuationsymbol}%
        \sbox\Wrappedvisiblespacebox {\FV@SetupFont\Wrappedvisiblespace}%
        \def\FancyVerbFormatLine ##1{\hsize\linewidth
            \vtop{\raggedright\hyphenpenalty\z@\exhyphenpenalty\z@
                \doublehyphendemerits\z@\finalhyphendemerits\z@
                \strut ##1\strut}%
        }%
        \def\FV@Space {%
            \nobreak\hskip\z@ plus\fontdimen3\font minus\fontdimen4\font
            \discretionary{\copy\Wrappedvisiblespacebox}{\Wrappedafterbreak}
            {\kern\fontdimen2\font}%
        }%
        
        \Wrappedbreaksatspecials
        \OriginalVerbatim[#1,codes*=\Wrappedbreaksatpunct]%
    }
    \makeatother

    \definecolor{incolor}{HTML}{303F9F}
    \definecolor{outcolor}{HTML}{D84315}
    \definecolor{cellborder}{HTML}{CFCFCF}
    \definecolor{cellbackground}{HTML}{F7F7F7}
    
    \makeatletter
    \newcommand{\boxspacing}{\kern\kvtcb@left@rule\kern\kvtcb@boxsep}
    \makeatother
    \newcommand{\prompt}[4]{
        \ttfamily\llap{{\color{#2}[#3]:\hspace{3pt}#4}}\vspace{-\baselineskip}
    }

\pgfplotsset{compat = newest}

\hypersetup{
  colorlinks   = true,
  linkcolor = red!70!black,
     citecolor    = green!50!black
}
\newtheorem{thm}{Theorem}[section]
\newtheorem{lem}[thm]{Lemma}
\newtheorem{conj}[thm]{Conjecture}

\newtheorem{cor}[thm]{Corollary}
\newtheorem*{thm*}{Theorem}
\newtheorem{qn}[thm]{Question}

\theoremstyle{definition}

\theoremstyle{remark}

\newtheorem*{acknowledgement}{Acknowledgments}

\newcommand{\de}{\ensuremath{\mathrm{deg}}}

\renewcommand{\le}{\leqslant}
\renewcommand{\leq}{\leqslant}
\renewcommand{\ge}{\geqslant}
\renewcommand{\geq}{\geqslant}

\def\longequation{$$\vcenter\bgroup\advance\hsize by -9em%
\noindent\ignorespaces\refstepcounter{equation}}%
\makeatletter%
\def\endlongequation{\egroup\eqno(\theequation)$$\global\@ignoretrue}
\makeatother

\begin{document}
\title[Distributed coloring
  and the local
  structure of unit-disk graphs]{Distributed coloring
  and the local
  structure \\of unit-disk graphs}
\author{Louis Esperet} \address{Laboratoire G-SCOP (CNRS,
  Grenoble-INP), Grenoble, France}
\email{louis.esperet@grenoble-inp.fr}

\author{S\'ebastien Julliot} \address{Laboratoire G-SCOP (CNRS,
  Grenoble-INP), Grenoble, France}

\author{Arnaud de Mesmay} \address{LIGM, CNRS, Univ. Gustave Eiffel, ESIEE Paris, F-77454 Marne-la-Vallée, France} \email{arnaud.de-mesmay@univ-eiffel.fr}

\thanks{The authors are partially supported by ANR Projects GATO
(\textsc{anr-16-ce40-0009-01}), GrR (\textsc{anr-18-ce40-0032}), MIN-MAX (\textsc{anr-19-ce40-0014}) and SoS (\textsc{anr-16-ce40-0009-01}), and by LabEx PERSYVAL-lab (\textsc{ANR-11-LABX-0025}). A preliminary version of this work appeared in the proceedings of the 17th International Symposium on Algorithms and Experiments for Wireless Sensor Networks (ALGOSENSORS 2021)~\cite{proc}.}

\date{}
\sloppy

\begin{abstract}
Coloring unit-disk graphs efficiently is an important problem in the global and
distributed setting, with applications in  radio channel assignment
problems when the communication relies on omni-directional antennas of
the same power. In this context it is important to bound not only the
complexity of the coloring algorithms, but also the number of colors
used. In this paper, we consider two natural distributed settings.  In
the location-aware setting (when nodes know their coordinates in the
plane), we give a constant time distributed algorithm coloring any
unit-disk graph $G$ with at most $4\omega(G)$ colors,  where
$\omega(G)$ is the clique number of $G$. This improves upon a classical
3-approximation algorithm for this problem, for all unit-disk graphs
whose chromatic number significantly exceeds their clique number. When
nodes do not know their coordinates in the plane, we give a
distributed algorithm in the \textsf{LOCAL} model that colors every
unit-disk graph $G$ with at most $5.68\,\omega(G)+1$ colors  in $O(\log^* n)$ rounds. This algorithm is based on
a study of the local structure of unit-disk graphs, which is of
independent interest. We conjecture that every unit-disk graph $G$ has
average degree at most $4\omega(G)$, which would imply the existence
of a $O(\log n)$
round algorithm coloring any unit-disk graph $G$ with (approximately) $4\omega(G)$ colors in the \textsf{LOCAL} model. We provide partial results towards this conjecture using Fourier-analytical tools.

\medskip

\end{abstract}
\maketitle

\section{Introduction}

A \emph{unit-disk graph} is a graph $G$ whose vertex set is a collection of
points $V\subseteq \mathbb{R}^2$, and such that two vertices $u,v\in
V$ are adjacent in $G$ if and only if  $\|u-v\|\le 1$, where $\|\cdot \|$
denotes the Euclidean norm.
Unit-disk graphs are a classical model of 
wireless communication networks, and are a central object of study in
distributed algorithms (see the survey~\cite{Suo13} for an extensive
bibliography on this topic). A classical way to design distributed communication protocols avoiding
interferences is to find a proper coloring of the underlying unit-disk
graph: the protocol then lets each vertex of the first color
communicate with their neighbors, then each vertex of the second
color, etc. Clearly the efficiency of the protocol depends on the
number of colors used, so it is important to minimize the total number of
colors (in addition to optimizing the complexity of the distributed coloring
algorithm).

The \emph{chromatic number} of a graph $G$, denoted by $\chi(G)$, is the
smallest number of colors in a proper coloring of $G$. The \emph{clique
number} of $G$, denoted by $\omega(G)$, is the largest size of a clique
(a set of pairwise adjacent vertices) in $G$. Note that for any graph
$G$ we have $\omega(G)\le \chi(G)$, but the gap between the parameters
can be arbitrarily large  in general (see~\cite{SS20} for a recent
survey on the relation between $\omega$ and $\chi$ for various graph classes). However, for unit-disk
graphs it is known that $\chi(G)\le 3\omega(G)-2$ \cite{Pee91} (see
also~\cite{GSW98} for a different proof), and improving
the multiplicative constant 3 is a longstanding open problem. It
should be noted that while computing $\chi(G)$ for a unit-disk graph $G$ is
NP-hard, computing $\omega(G)$ for a unit-disk graph $G$ can be done in
polynomial time~\cite{CCJ90}.

\medskip

In the \textsf{LOCAL} model, introduced by Linial~\cite{Lin92}, the
graph $G$ that we are trying to color models a
communication network: its vertices are processors of infinite
computational power and its edges are communication links between
(some of) these nodes. The vertices exchange messages with their neighbors in a certain
number of synchronous
rounds of communication (the round complexity), and then (in the case of graph coloring) each vertex
outputs its color in a proper coloring of $G$.
When in addition each vertex knows its coordinates in the plane, we
call the model the \textsf{location-aware} \textsf{LOCAL} model (more
details about these models will be given in Section~\ref{sec:local}).

\medskip

In the \textsf{location-aware} \textsf{LOCAL} model, the following is a classical
result~\cite{LA1,LA2,LA3,LA4} (see also~\cite{Suo13} for a survey on
\emph{local algorithms}, which are algorithms that run in a constant
number of rounds).

\begin{thm}[\cite{LA1,LA2,LA3,LA4}]\label{thm:3approx}
  A coloring of any unit-disk
graph $G$ with at most $3 \chi(G)$ colors can be obtained in a constant number
of rounds by a deterministic distributed algorithm in the \textsf{location-aware} \textsf{LOCAL} model.
\end{thm}

We prove the
following complementary result, which improves the number of colors as
soon as  $\chi(G)\ge \tfrac43 \omega(G)$. Note that there exists an infinite family of  unit disk graphs $G$ for which $\chi(G)\ge \tfrac32 \omega(G)$~\cite{MPW98}. Our algorithm is inspired by the
proof of~\cite{GSW98} showing that unit-disk graphs satisfy $\chi\le 3\omega$.

\begin{restatable}{thm}{geom}\label{thm:geom}
A coloring of any unit-disk graph $G$ with at
most $4\,\omega(G)$ colors can be computed in a constant number of rounds by a deterministic algorithm in the \textsf{location-aware} \textsf{LOCAL} model.
\end{restatable}

%\begin{thm}\label{thm:geom}
%A coloring of any unit-disk graph $G$ with at
%most $4\omega(G)$ colors can be computed in a constant number of rounds by a deterministic algorithm in the \textsf{location-aware} \textsf{LOCAL} model.
%\end{thm}

%As mentioned above we have $\chi(G)\ge \omega(G)$ for any
%graph $G$, so this result improves upon \cref{thm:3approx} as soon as $\omega(G)$ is sufficiently large and $\chi(G)\ge (1+\delta)\omega(G)$, for some
%fixed $\delta>0$. It was proved that for different models of random
%unit-disk graphs the ratio between the chromatic number and the clique
%number is equal to $\tfrac{2\sqrt{3}}{\pi}\approx 1.103$ with high probability~\cite{McD03,MR99}, so this shows
%that \cref{thm:geom} outperforms \cref{thm:3approx} for almost all
%unit-disk graphs (with respect to these
%distributions).

%Note that this ratio of $\tfrac{2\sqrt{3}}{\pi}$
%naturally corresponds to the case of points that are uniformly
%distributed in the plane.

%\medskip

\medskip

Given two integers $p \ge q \ge 1$, a \emph{$(p\!:\!q)$-coloring} of a
graph $G$ is an assignment of $q$-element subsets of $[p]$ to the
vertices of $G$, such that the sets assigned to any two adjacent
vertices are disjoint. 
The \emph{fractional chromatic number} $\chi_f(G)$ is
defined as the infimum of $\{\tfrac{p}{q}\,|\, G \mbox{ has a }
(p\!:\!q)\mbox{-coloring}\}$~\cite{SU13} (it can be proved that this
infimum is indeed a minimum). Observe that a
$(p\!:\! 1)$-coloring is a (proper) $p$-coloring, and
that for any graph $G$, $\omega(G)\le \chi_f(G)\le \chi(G)$. The
fractional chromatic number is often used in scheduling as an
alternative to the chromatic number when resources are fractionable,
which is the case for communication protocols. It was proved in \cite{GM01}
that for any unit-disk graph $G$, $\chi_f(G)\le 2.155\,
\omega(G)$. Here we give an efficient distributed implementation of
this result.

\begin{restatable}{thm}{geomfrac}\label{thm:geomfrac}
There is a constant $q\in \mathbb{N}$, such that in any unit-disk graph $G$, there exists an integer $p$ with $\tfrac{p}q\le 2.156\,
\omega(G)$ such that a
$(p\!:\!q)$-coloring of $G$   can be computed in $O(1)$
rounds by a deterministic distributed algorithm in the \textsf{location-aware} \textsf{LOCAL} model.
\end{restatable}

We now turn to the abstract setting, where vertices do not have access to
their coordinates in the plane. For a real number $n>0$, let $\log^* n$ be the number of times we have to
    iterate the logarithm, starting with $n$, to reach a value in
    $(0,1]$. Since paths are unit-disk graphs and  coloring
$n$-vertex paths
with a constant number of colors takes $\Omega(\log^* n)$ rounds in
the \textsf{LOCAL} model~\cite{Lin92}, coloring unit-disk graphs of bounded clique
number with a bounded number of colors also takes $\Omega(\log^* n)$ rounds in
the \textsf{LOCAL} model. Recalling that for any unit-disk graph $G$,
$\omega(G)\le \chi(G)\le 3\omega(G)$, a natural question is the
following.

\begin{qn}\label{q:mincol}
What is the minimum real $c>0$ such that a coloring of any $n$-vertex
unit-disk graph $G$ with $c \cdot \omega(G)$ colors can be obtained in
$O(\log^* n)$ rounds in the \textsf{LOCAL} model?
\end{qn}

Using the folklore result that any
unit-disk graph $G$ has maximum degree at most
$6\omega(G)-6$ (see~\cite{GSW98}), together with the fact that unit-disk graphs of maximum degree $\Delta$ can be colored efficiently with $\Delta+1$ colors in the \textsf{LOCAL} model~\cite{SW10}, we deduce that unit-disk graphs $G$ can
be colored efficiently with $6\omega(G)$ colors in the \textsf{LOCAL} model.
We obtain the following improved version by studying the local
structure of unit-disk graphs, using techniques that might be of
independent interest.

\begin{restatable}{thm}{twocol}\label{thm:2col}
Every unit-disk graph $G$ can be colored with at
most $5.675\,\omega(G)+1$ colors by a deterministic distributed algorithm in the \textsf{LOCAL} model,
running  in $O(\log^*n)$ rounds.
\end{restatable}

In relation to \cref{q:mincol}, it is natural to study the power of
graph coloring algorithms in unit-disk graphs in a
different (less restrictive) range of 
round complexity.

\begin{qn}\label{q:mincollog}
What is the minimum real $c>0$ such that a coloring of any $n$-vertex
unit-disk graph $G$ with $c \cdot \omega(G)$ colors can be obtained in
$O(\log n)$ rounds in the \textsf{LOCAL} model?
\end{qn}

An interesting property of the $O(\log n)$ range of round complexity
(compared to the $O(\log^* n)$ range)
is that it allows to solve coloring problems for graphs of bounded
average degree (rather than bounded maximum degree).
The \emph{average degree} of a graph $G$ is the average of its vertex degrees.
The \emph{maximum average degree} of a graph $G$ is the maximum
average degree of a subgraph $H$ of $G$.
In~\cite{BE10}, Barenboim and Elkin gave, for any $\epsilon>0$, a
deterministic distributed algorithm coloring $n$-vertex graphs of
maximum average degree $d$ with at most $(1+\epsilon)d+3$ colors in
$O(\tfrac{d}{\epsilon} \log n)$ rounds (the result was proved in terms
of arboricity rather than average degree).

\medskip

While the chromatic number and degeneracy of unit-disk graphs (as a
function of the clique number) are well studied topics, it seems 
that little is known about the average degree of unit-disk graphs.
We conjecture the following:

\begin{conj}\label{conj}
Every unit-disk graph $G$ has average degree at most $4\,\omega(G)$.
\end{conj}

It can be checked that the constant 4 is best possible by
considering uniformly distributed points in the plane. In this case
each vertex has degree equal to some density constant $c>0$ times the
area of a disk of radius 1, so the average degree is $c\cdot \pi$. On
the other hand any clique is contained in a region of diameter at most
1, and the area of such a region is known to be maximized for a disk
of radius $\tfrac12$~\cite{Bie15}, i.e., the graph has clique number 
$c\cdot \pi/4$, giving a ratio of 4 between the average degree and the
clique number.

\medskip

Using the result of Barenboim and Elkin \cite{BE10} mentioned above,
\cref{conj} would imply the existence of a deterministic distributed coloring
algorithm using $(4+\epsilon)\,\omega(G)$ colors in
$O(\tfrac{\omega(G)}{\epsilon} \log n)$ rounds (for fixed $\epsilon>0$ and
sufficiently large $\omega(G)=\Omega(1/\epsilon$)).

\medskip

Unfortunately we are quite far from proving \cref{conj} at the
moment. Our best result so far is the following.

\begin{restatable}{thm}{ad}\label{thm:ad}
Every unit-disk graph $G$ has average degree at most $5.68\,\omega(G)$.
\end{restatable}

Our final result shows that \cref{conj}, if true, is more subtle that it might seem. In the example above showing the optimality of \cref{conj}, the largest cliques are formed by sets of points that are all contained within some disk of radius $\frac{1}{2}$.  We may naturally define the \emph{disk clique number} of a unit-disk graph $G$, denoted by $\omega_D(G)$, to be the largest size of a clique contained within such a disk. Note that $\omega_D(G)$ depends on the embedding of $G$ in the plane, not just of the underlying graph (contrary to the clique number $\omega(G)$). One may wonder whether \cref{conj} holds for $\omega_D(G)$ instead of $\omega(G)$ -- by the same argument as for \cref{conj}, the constant $4$ would be best possible. It is however not the case and we show the following stronger lower bound:

\begin{restatable}{thm}{fourier}\label{thm:fourier}
There exists a unit-disk graph $G$ of average degree at least $4.0905\, \omega_D(G)$.
\end{restatable}

Here, $4.0905$ is roughly $4 \cdot (1+\frac{J_1(2B)}{2B})$, where $J_1$ is the Bessel function of the first kind and $B$ its first zero. \cref{thm:fourier} shows that any approach to prove \cref{conj} needs to account for possible shapes of cliques different than the ones contained in disks of radius $\frac{1}{2}$. As the constant suggests, our proof technique relies on Fourier analysis and special functions and we believe that they could be of independent interest in the study of unit-disk graphs, and more generally of intersection graphs of other objects in the plane. In particular they also allow us to prove \cref{conj} when the distribution of points is sufficiently close to the uniform distribution (in the sense that the support of its Fourier transform is contained in a small disk around $0$).

\subsection*{Organization of the paper} We start with a
presentation of the \textsf{LOCAL} model and some basic results on
coloring and 
unit-disk graphs in \cref{sec:prelim}. \cref{sec:coco} is devoted to
proving our main results in the location-aware setting,
\cref{thm:geom} and \cref{thm:geomfrac}. In \cref{sec:cowico}, we study
the local structure of unit-disk graphs and deduce our coloring
result in the \textsf{LOCAL} model, \cref{thm:2col}, together with our
upper bound on the average degree of unit-disk graphs, \cref{thm:ad} above.
In \cref{sec:fourier}, we introduce Fourier-analytical tools to approach \cref{conj} and prove \cref{thm:fourier}. %We prove in particular the
%conjecture when the distribution of points is sufficiently close from the
%uniform distribution (in the sense that the support of the Fourier transform is contained in a small disk around 0), but we also show the limits of our approach by establishing \cref{thm:fourier}.

\section{Preliminaries}\label{sec:prelim}

\subsection{Distributed models of communication}\label{sec:local}

All our results are
  proved in the \textsf{LOCAL} model, introduced by
  Linial~\cite{Lin92}. The underlying network is modelled as an $n$-vertex graph $G$ whose vertices have
  unbounded computational power, and whose edges are communication links
  between the corresponding vertices. In the
  case of deterministic algorithms, each
  vertex of $G$ starts with an arbitrary unique identifier (an integer between
  1 and $n^\alpha$, for some constant $\alpha\ge 1$, such that all integers
  assigned to the vertices are distinct). For randomized algorithms, each vertex
  starts instead with a collection of (private) random bits. The vertices then exchange messages
  (possibly of unbounded size)
  with their neighbors in synchronous rounds, and after a fixed number
  of rounds (the \emph{round complexity of the algorithm}), each
    vertex $v$ outputs its local ``part'' of a global solution to a combinatorial
    problem in $G$, for instance its color $c(v)$ in some
    proper $k$-coloring $c$ of $G$. 

    It turns out that with the assumption that
    messages have unbounded size, after $t$ rounds we can assume
    without loss of generality that each vertex $v$ ``knows'' its
    neighborhood $B_t(v)$ at distance $t$ (the set of all vertices at
    distance at most $t$ from $v$). More specifically $v$ knows the
    labelled subgraph of $G$ induced by $B_t(v)$ (where the labels are
    the identifiers of the vertices), and nothing more, and the
    output of $v$ is based solely on this information (see~\cite{Lin92}).

    The goal is to minimize the round complexity. Since nodes have
    infinite computational power, the paragraph above shows that any
    problem can be solved in a number of rounds equal to the diameter
    of the graph, which is at most $n$ when $G$ is connected.  The
    goal is to obtain algorithms that are significantly more
    efficient, i.e., of round complexity $O(\log n)$, or even
    $O(\log^* n)$.

    \medskip

In this paper, we will also consider  the \textsf{location-aware LOCAL} model, which is a variant of the \textsf{LOCAL} model in which  the $n$-vertex graph  modelling the
    communication network is a unit-disk graph embedded in the
    plane, and every vertex knows its coordinates in the embedding.

\subsection{Distributed coloring}\label{sec:localcol}

Consider a graph $G$. In the \emph{$(\de+1)$-list coloring problem}, each vertex
$v$ is given a list $L(v)$ of colors such that $|L(v)|\ge
d(v)+1$, where $d(v)$ denotes the degree of $v$ in $G$, and the goal is to color each vertex $v$ with a
color from its list $L(v)$, so that any two adjacent vertices receive different
colors. 

%The union of all the lists in a list coloring problem is called
%\emph{the color space}, and in some results below it will be
%convenient to assume that it has bounded size.
%The following result was obtained recently %in~\cite{HKNT21}.

%\begin{thm}[\cite{HKNT21}]\label{thm:lcolp}
%There  exists
%a randomized distributed algorithm in the \textsf{LOCAL} model  that solves the
%$(\de+1)$-list coloring problem in $n$-vertex graphs in  $O(\log ^3 \log n)$ rounds w.h.p.
%\end{thm}

A graph $G$ has \emph{bounded growth} if there is a function $f$ such that for any integer $r\ge 1$, and any vertex $v\in V(G)$, the maximum number of independent (i.e., pairwise non-adjacent) vertices among the vertices at distance at most $r$ from $v$ is at most $f(r)$. 
It was proved in \cite{SW10} that any graph of bounded growth and maximum degree $\Delta$ can be colored with $\Delta+1$ colors in $O(\log^* n)$ rounds in the \textsf{LOCAL} model (where the hidden constant in the $O(\cdot)$ notation only depends on the function $f$ bounding the growth). As noted by an anonymous reviewer, the proof of  \cite{SW10}  extends verbatim to the $(\de+1)$-list coloring problem.

\begin{thm}[\cite{SW10}]\label{thm:lcolp1}
There  exists
a deterministic distributed algorithm in the \textsf{LOCAL} model  that solves the
$(\de+1)$-list coloring problem in any $n$-vertex graph of bounded growth in  $O(\log^* n)$ rounds.
\end{thm}

It can be checked that that unit-disk graphs have bounded growth (the function $f$ in this case is quadratic, see~\cite{SW10}), which yields the following immediate
corollary.

\begin{cor}\label{cor:lcolp1}
There  exists
a deterministic distributed algorithm in the \textsf{LOCAL} model  that solves the
$(\de+1)$-list coloring problem in any $n$-vertex unit-disk graph in $O(\log^* n)$ rounds.
\end{cor}

\subsection{Unit-disk graphs}

For $0\le r \le 1$, and a point $v$ let $D_r(v)$ be the disk of radius
$r$ centered in $v$ and let $C_r(v)$ the circle of radius $r$ centered in
$v$. Given a unit-disk graph $G$ embedded in the plane, a vertex $v$, and a
real $0\le r \le 1$, we denote by $d_r(v)$ the number of neighbors of
$v$ lying in $D_r(v)$, and by $x_r(v)$ the number of neighbors of $v$
lying on $C_r(v)$. Note that
$d_1(v)$ is precisely $d(v)$, the degree of $v$ in $G$, and
for every $0\le r \le 1$, $d_r(v)=\sum_{s\in [0,r]} x_s(v)$ (note that since $s\mapsto x_s(v)$ has finite support, this sum is well defined).

It is well known that a disk of radius 1 can be covered by 6 regions of
diameter 1 (see~\cite{GSW98}), and thus the neighborhood of each vertex of $G$ can be covered by
6 cliques (and thus $d(v)\le 6 \omega(G)$ for each vertex $v$ of
$G$). A \emph{Reuleaux triangle} $R$ is the intersection of 3 disks of radius 1, centered in the three vertices of an equilateral triangle $T$ of side length 1 (see the green region in Figure~\ref{fig:nei}). The three vertices of $T$ are also called vertices of $R$. Note that Reuleaux triangles have diameter 1, and 6 Reuleaux triangles are enough to cover a disk of radius 1. Moreover, points close to the center of the disk are covered by more triangles than points on the outer circle. This can be used to prove the following.

\begin{figure}[htb]
 \centering
 \includegraphics[scale=0.8]{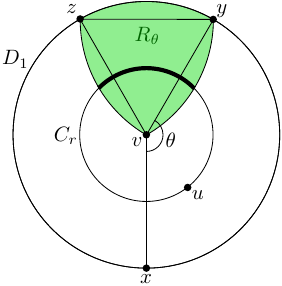}
 \caption{The intersection (in bold) of a circle $C_r$ of radius $r$ with a
   Reuleaux triangle $R_\theta$ (the green region).}
 \label{fig:nei}
\end{figure}

\begin{lem}\label{lem:radius12}
For each vertex $v$ of a unit-disk graph $G$ embedded in the plane,
$\sum_{r\in [0,1]} (2-r)x_r(v) \le 6\omega(G)$.
\end{lem}

\begin{proof}
  Let $v$ be a vertex of $G$. Note that each neighbor of $v$ lies in $D_1(v)$ (the disk of radius 1
centered in $v$). Let $x$ be an arbitrary point of $C_1(v)$. For  $\theta\in \mathopen[0,2\pi\mathclose)$, let $R_\theta$ be the Reuleaux triangle having $y,v,z$ as a vertices, in clockwise order  (with $y$ and $z$ two vertices of $C_1(v)$ at distance 1), and such that the angle $\angle xvy$ is equal to $\theta$ ($R_\theta$ is depicted in green in \cref{fig:nei}). 
Consider also a neighbor $u$ of $v$
lying on $C_r(v)$, the circle of radius $r$ centered in $v$. For
$\theta$ chosen uniformly at random in the interval $\mathopen[0,2\pi\mathclose)$, the
probability that $u$ is covered by $R_\theta$ is the length of
$C_r(v)\cap R_\theta$ (the bold arc of \cref{fig:nei}) divided by
$2\pi r$ (the circumference of $C_r(v)$). This probability is thus $$\tfrac{(2
  \arccos(r/2)-\pi/3)r}{2 \pi r}=\tfrac1\pi \arccos(r/2)-1/6.$$
  Note that $x \mapsto \arccos x$ is concave in $[0,\tfrac12]$, and thus $r\mapsto f(r):=\tfrac1\pi \arccos(r/2)-1/6$ is concave in $[0,1]$. As a consequence, for any $r\in [0,1]$, $f(r)\ge (f(1)-f(0))r+f(0)=(\tfrac16-\tfrac13)r+\tfrac13=\tfrac13-\tfrac{r}6$.
  
It follows that each neighbor of $v$ at distance at most $r$ from $v$
is covered by $R_\theta$ with probability at least
$\tfrac13-\tfrac{r}6$. Therefore, the expected number of neighbors of $v$
covered by $R_\theta$ is at least $\sum_{r\in [0,1]} (\tfrac13-\tfrac{r}6)x_r(v)$. Since each $R_\theta$ has diameter 1, the
vertices of $R_\theta$ 
induce a clique in $G$ and thus $R_\theta$ contains at most $\omega(G)$
vertices. It follows that $\sum_{r\in [0,1]} (\tfrac13-\tfrac{r}6)x_r(v)\le \omega(G)$
and thus $\sum_{r\in [0,1]} (2-r)x_r(v) \le 6\omega(G)$, as desired.
\end{proof}

The following is a direct consequence of \cref{lem:radius12}.

\begin{cor}\label{cor:radius12}
For each vertex $v$ of a unit-disk graph $G$ embedded in the plane and
$r\in [0,1]$,
$d(v)+(1-r)d_r(v) \le 6\omega(G)$. In particular $d_{1/2}(v)\le 12\omega(G)-2d(v)$.
\end{cor}

\begin{proof}
By \cref{lem:radius12}, $6\omega(G)\ge \sum_{s\in [0,1]} (2-s)x_s(v)=\sum_{s\in [0,r]} (2-s)x_s(v)+\sum_{s\in [r,1]}
(2-s)x_s(v)\ge (2-r)d_r(v)+d(v)-d_r(v)=(1-r)d_r(v)+d(v)$. By taking
$r=\tfrac12$, we obtain $d(v)+\tfrac12 d_{1/2}(v)\le 6\omega(G)$, and
thus $d_{1/2}(v)\le 12\omega(G)-2d(v)$, as desired.
\end{proof}

In the next section will  need the following useful observation of~\cite{GSW98} (see also~\cite{GM01}).

\begin{lem}[\cite{GSW98}]\label{lem:coco}
Let $G$ be a unit-disk graph embedded in the plane, such that the
$y$-coordinates of any two vertices of $G$ differ by at most
$\sqrt{3}/2$. Then $\chi(G)=\omega(G)$.
\end{lem}

We deduce the following easy corollary (which holds in the
\textsf{LOCAL} model, so it does not require nodes to know their own
coordinates in the plane).

\begin{cor}\label{cor:coco}
Let $G$ be a unit-disk graph embedded in the plane, such that the
$y$-coordinates of any two vertices of $G$ differ by at most
$\sqrt{3}/2$, and the $x$-coordinate differ by at most $\ell$, for
some real number $\ell$.  Then $G$ can be colored with
$\chi(G)=\omega(G)$ colors by a deterministic distributed
algorithm running in $O(\ell)$ rounds in the \textsf{LOCAL} model.
\end{cor}

\begin{proof}
  Take any shortest path $P=v_1,v_2,\ldots, v_k$ in $G$. Then for any
  $1\le i \le k-2$, $\|v_i-v_{i+2}\|>1$, since otherwise   $P$ would
  not be a shortest path in $G$. Since $y$-coordinates differ by at
  most $\sqrt{3}/2$, it follows that the $x$-coordinates of $v_i$ and
  $v_{i+2}$ differ by at least $\tfrac12$, and the $x$-coordinates of
  the vertices $v_i$ with $i$ odd are monotone (say increasing with
  loss of generality). Hence, the $x$-coordinates of $v_1$ and $v_k$
  differ by at least $k/4$, and by definition, $k\le 4\ell$. As a consequence,
any connected component of $G$ has diameter at most $4\ell$, and thus
any connected component can be colored optimally by a deterministic distributed
algorithm running in $O(\ell)$ rounds in the \textsf{LOCAL} model.
\end{proof}

\section{Location-aware coloring}\label{sec:coco}

In the proceedings version of this paper, we claimed that for any $0<\epsilon\le 1$, a coloring of any unit-disk graph $G$ with at
most $(3+\epsilon)\omega(G)+6$ colors can be computed in $O(1/\epsilon)$
rounds by a deterministic distributed algorithm in the \textsf{location-aware} \textsf{LOCAL} model \cite[Theorem 8]{proc}. Unfortunately, we recently discovered an error in our proof of \cite[Lemma 3]{proc}, on which \cite[Theorem 8]{proc} was based, and at the moment we are not able to fix the proof. Here we prove the following weaker version instead.

\begin{figure}[htb]
 \centering
 \includegraphics{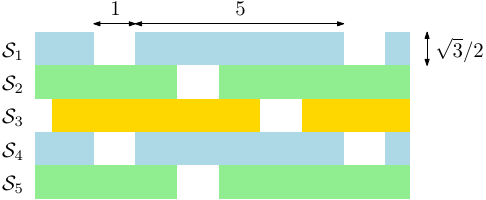}
 \caption{Covering $V(G)$ with horizontal stripes in the proof of
   \cref{thm:geom}.}
 \label{fig:stripes3}
\end{figure}

\geom*
%\begin{thm}\label{thm:geom2}
%A coloring of any unit-disk graph $G$ with at
%most $4\omega(G)$ colors can be computed in $O(1)$
%rounds by a deterministic distributed algorithm in the \textsf{location-aware} \textsf{LOCAL} model.
%\end{thm}

\begin{proof}
In the remainder, by the \emph{distance between two regions} $R_1$ and $R_2$ of the plane, we mean the minimum distance between a point of $R_1$ and a point of $R_2$.
We start by covering $V(G)$ with consecutive horizontal stripes
$\mathcal{S}_1,\mathcal{S}_2,\ldots$, each of height $\sqrt{3}/2$ (see
\cref{fig:stripes3}). Note that any two stripes $\mathcal{S}_j$ and $\mathcal{S}_{j+3}$
are at distance at least $2\cdot \sqrt{3}/2=\sqrt{3}>1$ apart, so
there are no edges connecting a vertex lying in $\mathcal{S}_j$ and a
vertex lying in $\mathcal{S}_{j+3}$. 

Let $\mathcal{R}$ be a union of rectangles of height $\sqrt{3}/2$ and length 1 (depicted in white in \cref{fig:stripes3}), with the following properties:
\begin{enumerate}
\item each rectangle of $\mathcal{R}$ is included in some horizontal stripe $\mathcal{S}_i$, $i\ge 1$,
\item any two consecutive rectangles of $\mathcal{R}$ on a stripe $\mathcal{S}_i$ lie at distance 5 apart, 
\item the distance between the ends of a stripe and the closest rectangle of $\mathcal{R}$ on that stripe is at most 5. 
\item any two rectangles of $\mathcal{R}$ lie at distance more than 1 apart
\end{enumerate}

Note that (4)  can be obtained by simply shifting the rectangles of $\mathcal{R}$ in $\mathcal{S}_{i}$ to the right to obtain the rectangles of $\mathcal{R}$ in $\mathcal{S}_{i+1}$ (and possibly adding a new rectangle of $\mathcal{R}$ at the beginning of $\mathcal{S}_{i+1}$).

For any $i\ge 1$, let $\mathcal{S}_i'=\mathcal{S}_i\setminus \mathcal{R}$. Note that each  $\mathcal{S}_i'$ consists of a sequence of rectangles of height $\sqrt{3}/2$ and length at most 5, all at distance more than 1 apart. 
It follows from this observation
and \cref{cor:coco} that for each $i \in \{0,1,2\}$, the subgraph $G_i$ of $G$ induced by the vertices lying in the union of all
$\mathcal{S}_j'$ with $j \equiv i \pmod 3$ has
chromatic number at most $\omega:=\omega(G)$, and can be colored with $\omega$ colors in $O(1)$ rounds in the \textsf{location-aware} \textsf{LOCAL} model. We will color each of these 3
graphs $G_i$, $i \in \{0,1,2\}$, with a disjoint set of at most $\omega$ colors. It remains to colors the vertices of $G$ lying in $\mathcal{R}$. Using \cref{cor:coco}, the set of vertices lying in each rectangle of $\mathcal{R}$ can be colored with $\omega$ colors in $O(1)$ rounds, and we can use the same set of $\omega$ colors for each rectangle of $\mathcal{R}$, as all these rectangles are at distance more than 1 apart. In total we have used at most $3\cdot \omega+\omega=4\omega$ colors, as desired.
\end{proof}

It was proved in \cite{GM01}
that for any unit-disk graph $G$, $\chi_f(G)\le 2.155\,
\omega(G)$. 
We now give an efficient distributed version of this result. 

\begin{figure}[htb]
 \centering
 \includegraphics{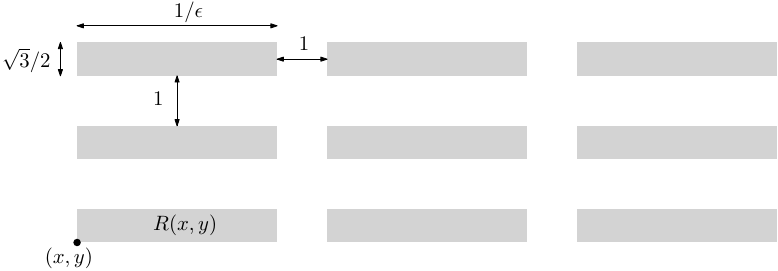}
 \caption{In grey, the region $\mathcal{R}(x,y)$ defined in the proof of
   \cref{thm:geomfrac}.}
 \label{fig:stripes}
\end{figure}

To give an intuition of the proof, note that the fractional chromatic number of $G$ can be equivalently defined as the infimum $x>0$ such that there is a probability distribution over the independent sets of $G$, such that with respect to this distribution, each vertex has probability at least $1/x$ to be in a random independent set. A way to define such a good distribution in a unit-disk graph $G$ is to choose a random point $(x,y)$, and then consider the grey rectangles depicted in \cref{fig:stripes}. The union of all vertices lying in grey rectangles induces a union of  perfect graphs of bounded diameter and can thus be colored in a constant number of rounds (in the \textsf{location-aware LOCAL} model) with $\omega(G)$ colors. We now take a random color class among the $\omega(G)$ color classes, and it can be easily checked that each vertex has probability at least $\tfrac{\sqrt{3/2}}{1+\sqrt{3/2}}\cdot  \tfrac{1/\epsilon}{1+1/\epsilon}\cdot \tfrac{1}{\omega(G)}\ge \tfrac1{(1+\epsilon)2.155\omega(G)}$ to lie in this random independent set (see \cite{GM01} for a version of this argument with horizontal stripes instead of rectangles). 

We now give an efficient distributed version of this argument that does not use randomization, and in particular allows to output $(p,q)$-colorings with bounded $p$ and $q$, which is a desirable property in the context of local algorithms~\cite{BKO21,BEP}.

\geomfrac*

%\begin{thm}\label{thm:geomfrac2}
%There is a constant $q\in \mathbb{N}$, such that in any unit-disk graph $G$, there exists an integer $p$ with $\tfrac{p}q\le 2.156
%\omega(G)$ such that a
%$(p\!:\!q)$-coloring of $G$   can be computed in $O(1)$
%rounds by a deterministic distributed algorithm in the \textsf{location-aware} \textsf{LOCAL} model.
%\end{thm}

\begin{proof}
Fix any real $\epsilon>0$. Let $G$ be a unit-disk graph embedded in the plane, and let $\omega=\omega(G)$. 
Given $(x,y)\in \mathbb{R}^2$, let $R(x,y)$ be the axis-parallel rectangle of height $\sqrt{3}/2$, length $1/\epsilon$, and bottom-left corner $(x,y)$.  Let $\mathcal{R}(x,y)$ be the union of all rectangles $R(x',y')$, for $x'\in x+(1+1/\epsilon)\cdot \mathbb{Z}$ and $y'\in  y+(1+\sqrt{3}/2)\cdot \mathbb{Z}$ (see the grey area in \cref{fig:stripes}).
It will be convenient to assume that
each rectangle contains its top  and right boundaries and excludes its bottom and left boundaries. In particular, any two points from distinct rectangles of $\mathcal{R}(x,y)$ are at distance more than 1 apart. Thus 
it follows from \cref{cor:coco} that for any $(x,y)\in \mathbb{R}^2$, the subgraph of $G$ induced by the vertices lying in $\mathcal{R}(x,y)$ can be colored with at most $\omega$ colors in $O(1/\epsilon)$ rounds in the \textsf{location-aware} \textsf{LOCAL} model.

\medskip

Fix some integer $r\ge 1$.
For any $0\le i \le r-1$, we consider the region $\mathcal{R}_i:=\mathcal{R}(\tfrac{i}{r}(1+1/\epsilon),\tfrac{i}{r}(1+\sqrt{3}/2))$. Note that for each point $(x',y')$ of the plane, there are at most $\lceil r\cdot \tfrac{1}{1+1/\epsilon}\rceil$ integers $0\le i\le r-1$ such that $x'$ lies outside of the projection of $\mathcal{R}_i$ on the horizontal axis, and there are at most $\lceil r\cdot \tfrac{1}{1+\sqrt{3}/2}\rceil$ integers $0\le i\le r-1$ such that $y'$ lies outside of the projection of $\mathcal{R}_i$ on the vertical axis. As a consequence, each point of the plane lies in at least 
\begin{eqnarray}
    r-\lceil r\cdot \tfrac{1}{1+\sqrt{3}/2}\rceil -\lceil r \cdot \tfrac{1}{1+1/\epsilon}\rceil & \ge & r-(r\cdot \tfrac{1}{1+\sqrt{3}/2}+1)-(r \cdot \tfrac{1}{1+1/\epsilon}+1)   \\
    & = & r-r\cdot \tfrac{1}{1+\sqrt{3}/2}-r \cdot \tfrac{1}{1+1/\epsilon}-2 \label{eq1}
\end{eqnarray}
regions $\mathcal{R}_i$. Note that for any $\epsilon>0$, $\tfrac{1}{1+1/\epsilon}=\tfrac{\epsilon}{\epsilon+1}\le \epsilon$, and thus \eqref{eq1} is at least $(0.464101-\varepsilon)\, r -2$. It follows that for sufficiently large (but constant) $r$  and sufficiently small (but constant) $\epsilon>0$, the quantity \eqref{eq1} is at least $0.4641r$.
Now, for any $0\le i \le r-1$, we consider a set $S_i$ of $\omega$ colors (such that all sets $S_i$ are pairwise disjoint) and color the subgraph of $G$ induced by the vertices lying in $\mathcal{R}_i$ with colors from $S_i$. By the paragraph above, this can be done in $O(1/\epsilon)=O(1)$ rounds in the \textsf{location-aware} \textsf{LOCAL} model. In total, at most $p:=r\omega$ colors are used, and since each point of the plane is covered by at least $0.4641 r$ regions $\mathcal{R}_i$, each vertex receives at least $q:=\lceil 0.4641 r\rceil$ different colors. We obtain a  $(p\!:\!q)$-coloring of $G$ with $p/q\le \tfrac{\omega}{0.4641}\le 2.156
\omega$, as desired.
\end{proof}

\section{Coloring without coordinates}\label{sec:cowico}

We start by giving a bound on the average of the degrees of two adjacent
vertices. Before proceeding with the proof of our result, we give a high-level intuition of our approach. A unit-disk graph can be viewed as multiset of points in the plane (several points might  be located at the same coordinates), or equivalently as a set of points carrying some integral positive weights, or equivalently again, as a function $\mu:\mathbb{R}^2\to \mathbb{N}^{\ge 0}$, where $\mu((x,y))=i\ge 0$ if and only if the unit-disk graph has $i$ vertices  located at $(x,y)$. Let us call $\mu((x,y))$ the \emph{weight} of the point $(x,y)$. The assumption that the underlying unit-disk graph has clique number at most $\omega$ can be described by an (infinite) set of \emph{linear} inequalities of the following form: for any region $R$ of the plane with diameter 1, the sum of the weights of the points of $R$ is at most $\omega$. The degree of a vertex $v$ can also be described as a linear combination of the weights: it is the sum of the weights of the points in $D_1(v)$. Therefore, the average of the degrees of two vertices $u$ and $v$ can also be described by a linear combination of the weights. We want to prove an upper bound on this quantity for all unit-disk graphs of maximum clique size $\omega$, so what we can do is set our weights as variables, add the (linear) constraints implying that all cliques have size at most $\omega$, and maximize the (linear) objective function defined as the average of the degrees of two adjacent vertices. The optimum will be the desired upper bound on the average of the degrees of two adjacent vertices in a unit-disk graph of maximum clique size $\omega$. However, we quickly run into a number of complications. 

\begin{itemize}
    \item In order to be solved efficiently, a linear program needs to allow for rational solutions. But this is fine, we can simply relax the condition that the weights are integral by allowing rational weights, and compute the optimum. It will still give an upper bound on the integral version of the linear program. 
    \item There is an infinite number of constraints, as there is an infinite number of regions of diameter 1 in the plane. So here too, we relax the problem and only consider a well chosen (finite) set of regions of diameter 1 and set as (linear) constraints that the sum of the weights in each of these regions is at most $\omega$. Again, the optimum with respect to this smaller set of inequalities will be at least the optimum for the original problem, and thus still be an upper bound on the average of the degrees of two adjacent vertices.
    \item There is an infinite number of variables (all points in $D_1(u)\cup D_2(v)$). Here, we argue that because of our choice of finitely many regions for the constraints, there are optimal solutions whose support is discrete (located at the intersection of the boundaries of these regions). Moreover, using symmetry arguments, we show that it suffices to consider a very small number of points, and eventually our proof boils down to solving two linear programs with a handful of variables and inequalities.
\end{itemize}

We now proceed with the formal proof of our result.
\begin{lem}\label{lem:ad2v}
For every two vertices $u$ and $v$ in a unit-disk graph $G$ embedded
in the plane such that $\tfrac12 \le \|u-v\|\le 1$, we have
$\tfrac12(d(u)+d(v))\le 5.675\,\omega(G)$.
\end{lem}

\begin{proof}
Let $\delta=\|u-v\|\in [\tfrac12,1]$ and $\omega=\omega(G)$. Given a bounded
subset $R$ of the plane, we denote by $\mu(R)$ the number of
vertices of $G$ lying in $R$. Note that for any vertex $x$ of $G$, $d(x)=\mu(D_1(x))$, and thus
it is enough to prove that
$\tfrac12(\mu(D_1(u))+\mu(D_1(v)))\le 5.675\,\omega(G)$. We will give two different upper
bounds on $\tfrac12(\mu(D_1(u))+\mu(D_1(v)))$, and the minimum of the two will be at
most $5.675\,\omega$ for every $\delta\in [\tfrac12,1]$. The second upper bound will be stronger than the first, except when  $\delta$, the distance between $u$ and $v$, is close to 1 (see Figure~\ref{fig:plot} for a comparison of the two upper bounds as a function of $\delta\in [\tfrac12,1]$).

  \begin{figure}[htb]
 \centering
 \includegraphics{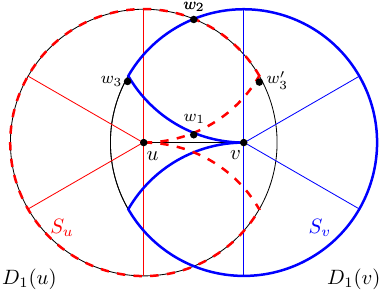}
 \caption{Two adjacent vertices $u$ and $v$ in the first part of the proof of \cref{lem:ad2v}.}
 \label{fig:2d23}
\end{figure}

\medskip

For the first bound, consider the region $S_u\subset D_1(u)$ (bounded
by the fat red dashed curve) and the region $S_v\subset D_1(v)$ (bounded
by the fat blue curve) in Figure \ref{fig:2d23}. 
$S_u$ is defined  as follows: take 6 points $z_1,\ldots,z_6$ appearing in clockwise order on $C_1(u)$, such that any two consecutive points $z_i,z_{i+1}$ (with indices modulo 6) form an equilateral triangle with $u$, and such that $z_1$ and $z_2$ are symmetric with respect to the line $(uv)$. Now define $S_u$ as the union of the five Reuleaux triangles with vertices  $uz_2z_3$, $uz_3z_4$, $uz_4z_5$, $uz_5z_6$, and $uz_6z_1$. The region $S_v$ is defined as the symmetric of $S_u$ with respect to the perpendicular bisector of the segment $[u,v]$.
Note that by definition, each
of $S_u$ and $S_v$ can be covered by 5 regions of diameter 1, and thus
\begin{eqnarray}\label{eq:1}
\mu(S_u)\le 5 \omega \quad \text{ and } \quad\mu(S_v)\le 5 \omega.
\end{eqnarray}
By
Lemma \ref{lem:radius12}, we also have
\begin{eqnarray}\label{eq:2}
  \sum_{r\in [0,1]} (2-r) \mu(C_r(u))   \le
  6\omega\quad \text{ and }\quad\sum_{r\in [0,1]} (2-r) \mu(C_r(v))  \le 6\omega
\end{eqnarray}
(recall that $C_r(u)$ is the circle of
radius $r$ centered in $u$). Note that $d(u)=\mu(D_1(u))
=\sum_{r\in [0,1]} \mu(C_r(u))  $ and  $d(v)=\mu(D_1(v))
=\sum_{r\in [0,1]} \mu(C_r(v))  $.

Subject to these
inequalities (and with $\omega$ fixed), our goal is to maximize
$\tfrac12(\mu(D_1(u))+\mu(D_1(v)))$ over all  point sets in the plane. As explained in the introduction of this section, it is convenient to relax the problem and optimize over points sets carrying some nonnegative (non necessarily integral) weights. In this context, $\mu(R)$ is simply defined as the sum of the weights of the points lying in $R$ (the objective function $\tfrac12(\mu(D_1(u))+\mu(D_1(v)))$ and the linear constraints \eqref{eq:1} and \eqref{eq:2}, which were defined using $\mu$, are modified accordingly). Thus we have a linear program with objective function $\tfrac12(\mu(D_1(u))+\mu(D_1(v)))$ (which we seek to maximize), whose variables are the weights of the points in the plane, with linear constraints given by \eqref{eq:1} and \eqref{eq:2}.
The solution of this linear
program will give us an upper bound for our original (integral)
optimization problem.

We now argue that some optimal solution of the linear program defined above has finite support. We define the following points (see Figure \ref{fig:2d23}): 
\begin{itemize}
\item $w_1$ is a point in $(D_1(u)\cap D_1(v))\setminus (S_u\cup S_v)$, 
%which is arbitrary close to $S_u\cap S_v$, and 
such that $\|u-w_1\|+\|v-w_1\|$ (the sum of the distances of $w_1$ to $u$ and $v$) is maximized,
\item $w_2$ is one of the two points in the intersection of $C_1(u)$ and $C_1(v)$,
\item $w_3$ is a point of $(C_1(v)\cap S_u)\setminus S_v$, which is arbitrarily close to $S_v$.
\item $w_3'$ is the symmetric of $w_3$ with respect to the perpendicular bisector of the segment $[u,v]$.
\end{itemize}

Observe that the linear inequalities \eqref{eq:1} and \eqref{eq:2} and the objective function $\tfrac12(\mu(D_1(u))+\mu(D_1(v)))$ are symmetric with respect to the line $(uv)$ and to the
perpendicular bisector of the segment $[uv]$.
This allows us to restrict ourselves to  optimal solutions that have the following additional properties:
\begin{itemize}
    \item they are symmetric with respect to the
perpendicular bisector of the segment $[uv]$ (i.e., we can interchange $u$ and $v$ without affecting the solution), and
\item their support lies in one of the two half-planes defined by $(uv)$, say the upper half-plane.
\end{itemize}

For the first property, it suffices to take the average of some optimal solution and the symmetric image of this solution with respect to the
perpendicular bisector of the segment $[uv]$, and for the second property, it suffices to take the union of the solution in the upper half-plane and the symmetric of the solution in the lower half-plane with respect to $(uv)$.

As $D_1(u)\setminus D_1(v)$ and $D_1(v)\setminus D_1(u)$ are symmetric with respect to the
perpendicular bisector of the segment $[uv]$, it follows that in some optimal solution as above we have $\mu(D_1(u)\setminus D_1(v))=\mu(D_1(v)\setminus D_1(u))$. If this quantity is non-zero, we can modify this solution by deleting all the points in the support of the solution that are in the symmetric difference of $D_1(u)$ and $D_1(v)$ and adding weight $\mu(D_1(u)\setminus D_1(v))=\mu(D_1(v)\setminus D_1(u))$ to the point $w_2$. The objective function remains unchanged and \eqref{eq:1} and \eqref{eq:2} are still satisfied, so we obtain a (symmetric) optimal solution whose support lies in the intersection of $D_1(u)$ and $D_1(v)$. Now,  if a symmetric weighted point set satisfies this property as well as \eqref{eq:1} and \eqref{eq:2}, and  contains a point $z$ (and its symmetric $z'$) distinct from
$w_1$, $w_2$, $w_3$ (and their symmetric images), then $z$ and $z'$ can be moved locally so that 
\begin{itemize}
\item the boundaries of $S_u$ and $S_v$ are not crossed in the motion, and
    \item $\|u-z\|+\|v-z\|$ increases (and by symmetry, $\|u-z'\|+\|v-z'\|$ increases).
\end{itemize}
The first property implies that the inequalities of \eqref{eq:1} remain valid after the motion. The second property implies that the left-hand sides of the two inequalities of \eqref{eq:2} are non-increasing: the contribution of $z$ and $z'$ to $\sum_{r\in [0,1]} (2-r) \mu(C_r(u))$ is \[(2-\|u-z\|)\,\mu(z)+(2-\|u-z'\|)\,\mu(z')=(4-(\|u-z\|+\|v-z\|))\,\mu(z),\] where we have used that $\mu(z)=\mu(z')$ and $\|u-z'\|=\|v-z\|$ by symmetry. The same holds  for the second inequality by symmetry. It follows that the two inequalities of \eqref{eq:2} are still valid after the motion, while the objective function remains unchanged. This shows that we can assume
without loss of generality that in some optimal solution, the support of the weighted point set is included in $\{w_1, w_2,w_3,w_3'\}$, and the weight
of $w_3$ is equal to the weight of $w_3'$.

\medskip

Once this has been observed, it remains to solve the following
finite  linear program (here $x_1=\mu(\{w_1\})$,
$x_2=\mu(\{w_2\})$, and $x_3=\mu(\{w_3\})=\mu(\{w_3'\})$).
\begin{gather*}
  \begin{aligned}
    \max_{x_1,x_2,x_3}\quad &x_1+x_2+2x_3  \\
\textup{s.t.}\quad &x_2 + x_3 \leq 5\omega \\
&\left(2-\sqrt{\tfrac{\delta^2}4+\left(1-\sqrt{1-\tfrac{\delta^2}4}\right)^2}\right)\cdot
                   x_1+x_2+\left(3-\sqrt{1+\delta^2-\delta\sqrt{3}}\right)\cdot
                   x_3\le 6\omega\\
         &x_1,x_2,x_3 \geq  0 \\
\end{aligned}
\end{gather*}
% \begin{maxi*}<b>
% {\scriptstyle x_1,x_2,x_3}{x_1+x_2+2x_3}{}{}
% \addConstraint {x_2 + x_3 \leq 5\omega}{}{}
% \addConstraint {\left(2-\sqrt{\tfrac{\delta^2}4+\left(1-\sqrt{1-\tfrac{\delta^2}4}\right)^2}\right)\cdot
%                    x_1+x_2+\left(3-\sqrt{1+\delta^2-\delta\sqrt{3}}\right)\cdot
%                    x_3\le 6\omega}{}{}
% \addConstraint {x_1,x_2,x_3  \geq  0},{}{}
% \end{maxi*}
where we have used
$d(u,w_1)=d(v,w_1)=\sqrt{\delta^2/4+(1-\sqrt{1-\delta^2/4})^2}$
and $d(u,w_3)=d(v,w_3')=\sqrt{1+\delta^2-\delta\sqrt{3}}$ (these expressions are obtained from repeated applications of Pythagoras' theorem starting from $d(u,v)=\delta$ and $d(u,w_3)=d(u,w_3')=d(u,w_3)=d(u,w_3')=1$, and using that $w_1$ is at distance 1 from two specific points of $C_1(u)$ and $C_1(v)$ forming a rectangle with $u$ and $v$). Let us
denote by $f_5(\delta)$ the optimum of this linear program. We compute
$f_5(\delta)$ for $\delta\in [\tfrac12,1]$ numerically using SageMath (see \cref{sec:comput}); the function $\tfrac{f_5(\delta)}\omega$ is plotted
in Figure~\ref{fig:plot} (red dashed curve).

\medskip

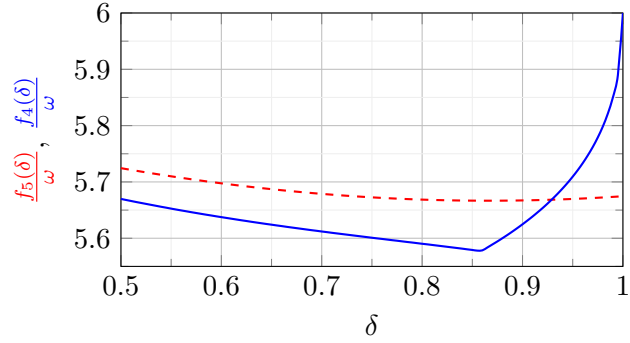
\begin{figure}[htb]
\begin{center}
\begin{tikzpicture}
\begin{axis}[
	xmin = 0.5, xmax = 1,
	ymin = 5.55, ymax = 6,
	xtick distance = 0.1,
	ytick distance = 0.1,
	grid = both,
	minor tick num = 1,
	major grid style = {lightgray},
	minor grid style = {lightgray!25},
	width = 0.5\textwidth,
	height = 0.3\textwidth,
	xlabel = {$\delta$},
	ylabel = {\textcolor{red}{$\tfrac{f_5(\delta)}\omega$}, \textcolor{blue}{$\tfrac{f_4(\delta)}\omega$}},]

% Plot data from a file
\addplot[
	thick, dashed,
	red
        ] file[skip first] {lp5.dat};

        \addplot[
	smooth,
	thick,
	blue
] file[skip first] {lp4.dat};

\end{axis}
\end{tikzpicture}
\end{center}
\caption{Two upper bounds on $\tfrac1{2\omega}(d(u)+d(v))$ when
  $\|u-v\|=\delta$. The function $\tfrac{f_5(\delta)}\omega$ is
  plotted with a red dashed curve, and the function $\tfrac{f_4(\delta)}\omega$ is
  plotted in blue.}
\label{fig:plot}
\end{figure}

We now turn to our second upper bound on
$\tfrac12(\mu(D_1(u))+\mu(D_1(v)))$. Consider the region $R_u\subset D_1(u)$ (bounded
by the fat red dashed curve) and the region $R_v\subset D_1(v)$ (bounded
by the fat blue curve) in Figure \ref{fig:2d23r}. The region $R_u$ is defined as follows:  take 6 points $z_1,\ldots,z_6$ appearing in clockwise order on $C_1(u)$, such that any two consecutive points $z_i,z_{i+1}$ (with indices modulo 6) form an equilateral triangle with $u$, and such that $z_1$ lies on the line  $(uv)$, and the angle $\angle z_1uv$ is equal to 0. Now define $R_u$ as the union of the four Reuleaux triangles with vertices $uz_2z_3$, $uz_3z_4$, $uz_4z_5$, and $uz_5z_6$. The region $R_v$ is defined as the symmetric of $R_u$ with respect to the perpendicular bisector of the segment $[u,v]$. By definition, 
each of the regions $R_u$ and $R_v$ can be covered by four regions of diameter 1, and thus
\begin{eqnarray}\label{eq:3}
\mu(R_u)\le 4 \omega \quad \text{ and } \quad\mu(R_v)\le 4 \omega.
\end{eqnarray}

 \begin{figure}[htb]
 \centering
  \includegraphics{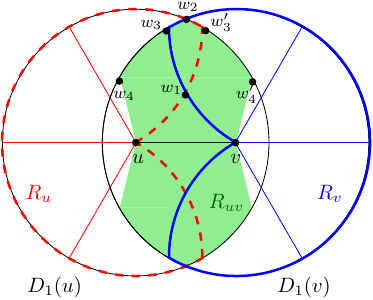}
 \caption{Two adjacent vertices $u$ and $v$ in the second part of the proof of \cref{lem:ad2v}.}
 \label{fig:2d23r}
\end{figure}

As before, by
Lemma \ref{lem:radius12}, we also have
\begin{eqnarray}\label{eq:4}
  \sum_{r\in [0,1]} (2-r) \mu(C_r(u))   \le
  6\omega\quad \text{ and }\quad\sum_{r\in [0,1]} (2-r) \mu(C_r(v))  \le 6\omega.
\end{eqnarray}
Note that $d(u)=\mu(D_1(u))
=\sum_{r\in [0,1]} \mu(C_r(u))  $ and  $d(v)=\mu(D_1(v))
=\sum_{r\in [0,1]} \mu(C_r(v))  $.

We now define a number of points and a new region (see Figure \ref{fig:2d23r}). Points $w_1,w_2,w_3,w_3'$ are defined similarly as in the first part of the proof, using $R_u$ and $R_v$ instead of $S_u$ and $S_v$.

\begin{itemize}
\item $w_1$ is a point in $(D_1(u)\cap D_1(v))\setminus (R_u\cup R_v)$, 
%which is arbitrary close to $R_u\cap R_v$,  and 
such that $\|u-w_1\|+\|v-w_1\|$ (the sum of the distances of $w_1$ to $u$ and $v$) is maximized,
    \item $w_2$ is one of the two points in the intersection of $C_1(u)$ and $C_1(v)$,
\item $w_3$ is a point of $(C_1(v)\cap R_u)\setminus R_v$, which is arbitrarily close to $R_v$,
\item $w_3'$ is the symmetric of $w_3$ with respect to the perpendicular bisector of the segment $[u,v]$,
\item $w_4\in C_1(v)$ and $w_4'\in C_1(u)$ are symmetric with respect to the perpendicular bisector of the segment $[u,v]$, and the distance between $w_4$ and $w_4'$ is equal to 1,
\item the region $R_{uv}$ (depicted in green  in Figure \ref{fig:2d23r}) is the union of triangles $uvw_4$ and $uvw_4'$,  circular sectors $w_4vw_2\subset D_1(v)$ and $w_2 u w_4'\subset D_1(u)$, and the symmetric region with respect to $(uv)$. For convenience we
assume that $R_{uv}$ does not contain $w_2$.
\end{itemize}

By definition, $d(u,w_4')=d(u,w_2)=d(w_4,w_4')=d(v,w_4)=1$, and thus each of the top and bottom halves of the region $R_{uv}$ has
diameter 1. Hence,
\begin{eqnarray}\label{eq:5}
\mu(R_{uv})\le 2 \omega.
\end{eqnarray}

The objective is again to maximize $\tfrac12(\mu(D_1(u))+\mu(D_1(v)))$ over all weighted point sets satisfying \eqref{eq:3}, \eqref{eq:4}, and \eqref{eq:5}.
As before, we can assume that some optimal solution is symmetric with respect to the perpendicular bisector of the segment $[uv]$, that its support lies in the upper half-plane defined by $(uv)$ and is included in the intersection of $D_1(u)$ and $D_1(v)$. If a symmetric weighted point set satisfies this property in addition to \eqref{eq:3}, \eqref{eq:4}, and \eqref{eq:5}, and contains two symmetric points $z$ and $z'$ distinct from
$w_1$, $w_2$, $w_3$, and $w_4$ (and their symmetric images), then $z$ and $z'$ can be moved locally so that $\|u-z\|+\|v-z\|$ and $\|u-z'\|+\|v-z'\|$ increase and thus no constraint is violated.  Together with the
symmetry of the constraints and objective function, this shows that we can assume
without loss of generality that in some optimal solution, the support of the weighted point set is included in $\{w_1,w_2,w_3,w_3',w_4,w_4'\}$
and moreover $\mu(\{w_3\})=\mu(\{w_3'\})$ and
$\mu(\{w_4\})=\mu(\{w_4'\})$. The optimization problem can thus again
be formulated as a finite linear program, as follows
(here $x_1=\mu(\{w_1\})$,
$x_2=\mu(\{w_2\})$, $x_3=\mu(\{w_3\})=\mu(\{w_3'\})$, and
$x_4=\mu(\{w_4\})=\mu(\{w_4'\})$).
\begin{gather*}
  \begin{aligned}
    \max_{x_1,x_2,x_3,x_4}\quad &x_1+x_2+2x_3+2x_4  \\
\textup{s.t.}\quad &x_2 + x_3 +x_4\leq 4\omega \\
&x_1 + 2x_3 \leq 2\omega  \\
&\left(2-\sqrt{\tfrac{\delta^2}4+\left(\tfrac{\sqrt{3}}2-\sqrt{1-\tfrac{(1+\delta)^2}4}\right)^2}\right)\cdot
  x_1+x_2+\left(3-\sqrt{1+\delta^2-\delta}\right)\cdot
  x_3\\
  &\qquad\qquad +\left(3-\sqrt{1-\delta}\right)\cdot
  x_4\le 6\omega\\
         &x_1,x_2,x_3,x_4  \geq  0, \\
\end{aligned}
\end{gather*}
% \begin{maxi*}<b>
% {\scriptstyle x_1,x_2,x_3,x_4}{x_1+x_2+2x_3+2x_4}{}{}
% \addConstraint {x_2 + x_3 +x_4\leq 4\omega}{}{}
% \addConstraint {x_1 + 2x_3 \leq 2\omega}{}{}
% \addConstraint
% {\left(2-\sqrt{\tfrac{\delta^2}4+\left(\tfrac{\sqrt{3}}2-\sqrt{1-(1+d)^2/4}\right)^2}\right)\cdot
%   x_1+x_2+\left(3-\sqrt{1+d^2-d}\right)\cdot
%   x_3}{}{}
% \addConstraint
% {}{\qquad+\left(3-\tfrac12\left(\sqrt{4-3\delta^2}-\delta\right)\right)\cdot
%   x_4\le 6\omega}{}
% \addConstraint {x_1,x_2,x_3  \geq  0},{}{}
% \end{maxi*}
where we have used
$d(u,w_1)=d(v,w_1)=\sqrt{\delta^2/4+\left(\sqrt{3}/2-\sqrt{1-(1+\delta)^2/4}\right)^2}$,
$d(u,w_3)=d(v,w_3')=\sqrt{1+\delta^2-\delta}$, and
$d(u,w_4)=d(v,w_4')=\sqrt{1-\delta}$ (these expressions are again obtained from repeated applications of Pythagoras' theorem). Let us
denote by $f_4(\delta)$ the optimum of this linear program. We compute
$f_4(\delta)$ for $\delta\in [\tfrac12,1]$ numerically using SageMath (see \cref{sec:comput}); the function $\tfrac{f_4(\delta)}\omega$ is plotted
in Figure~\ref{fig:plot} (blue curve).

We can now check that when $\delta\in [\tfrac12,1]$, $\min\{f_4(\delta),f_5(\delta)\}$ is maximized
when $\delta=1$, with $f_5(1)=5.6746\,\omega$ and $f_4(1)=6\omega$ (note
that $f_4(1/2)=5.6698\,\omega<f_5(1)$). It follows that
for any $u,v$ with $\|u-v\|\in [\tfrac12,1]$, $\tfrac12(d(u)+d(v))\le
5.6746\,\omega$, as desired.
\end{proof}

This bound easily implies that there is an efficient
distributed algorithm coloring $G$ with at most $5.675\,\omega(G)+1$
colors.

\twocol*

%\begin{thm}\label{thm:2col2}
%Every unit-disk graph $G$ can be colored with at
%most $5.68\,\omega(G)$ colors by a randomized distributed algorithm in the \textsf{LOCAL} model,
%running in $O(\log^3 \log n)$ rounds
%w.h.p. Moreover, if $\omega(G)=O(1)$, the coloring can be obtained deterministically in $O(\log^*n)$ rounds.
%\end{thm}

\begin{proof}
Let $\omega=\omega(G)$, let $A$ be the set of vertices of degree more than $5.675\,\omega$,
and let $B$ be the remaining vertices. We claim that any connected component of
$G[A]$, the subgraph of $G$ induced by $A$, is a clique. Indeed, any
two adjacent vertices in $G[A]$  are at distance at most $\tfrac12$ by
\cref{lem:ad2v} and thus if
$G[A]$ contains a path $uvw$, then $u$ and $w$ are at distance
$\tfrac12+\tfrac12\le 1$, and so $u$ and $w$ are adjacent. Since
$G[A]$ is a union of cliques, it can be colored with at most
$\omega$ colors in $O(1)$ rounds (each connected component has diameter at most 2, and contains at most $\omega$ vertices).
For each vertex $v\in B$, let $L(v)$ be the set of colors from $1,
2,\ldots , 5.675\,\omega+1$ that do not appear among the colored
neighbors of $v$. Let us denote by $d_A(v)$ and $d_B(v)$ the number of
neighbors of $v$ in $A$ and $B$, respectively. Note that for each
$v\in B$, $|L(v)|\ge 5.675\,\omega+1-d_A(v)\ge d_B(v)+1$, since  $d_A(v)+d_B(v)\le
5.675\,\omega$. Coloring the vertices of $B$ from their lists $L(v)$ is thus an
instance of the $(\de+1)$-list coloring
problem, which can be solved by a deterministic algorithm running in $O(\log^*n)$ rounds, by \cref{cor:lcolp1}. The resulting coloring is
a coloring of $G$ with at most $5.675\,\omega+1$ colors, as desired.
\end{proof}

As a direct application of  \cref{lem:ad2v}, we now obtain an improved upper
bound on the average degree of any unit-disk graph.

\ad*

%\begin{thm}\label{thm:ad2}
%Every unit-disk graph $G$ has average degree at most $5.68\,\omega(G)$.
%\end{thm}

\begin{proof}
Set $\epsilon=6-5.675=0.325$ (all the computations in the proof are with
respect to this specific choice of $\epsilon$) and $\omega=\omega(G)$, and consider a fixed
embedding of $G$ in the plane. By \cref{lem:ad2v}, the average of the degrees of any two vertices
$u,v$ with $\|u-v\|\in [\tfrac12,1]$ is most
$(6-\epsilon)\omega$.

Each vertex $v$ of $G$ starts with a charge
$w(v)=d(v)/\omega$, so that the average charge is precisely the average
degree of $G$ divided by $\omega$. The charge is then moved according to the following
rule: for any $\delta\ge 0$, each vertex with degree
$(6-\epsilon-\delta)\omega$ takes
$\tfrac{\delta}{(6-\epsilon-\delta)\omega}$ from the charge of each neighbor. For each $v\in
V(G)$, let $w'(v)$ be the resulting charge of $v$. Note that the total charge has not changed and thus the
average of $w'(v)$ over $v\in V$ is still the average
degree of $G$ divided by $\omega$. We now prove that $w'(v)\le
6-\epsilon+0.005=5.68$ for each vertex $v\in V$, which directly implies that $G$
has average degree at most $5.68\,\omega$.

Consider first a vertex $v$ of degree at most
$(6-\epsilon)\omega$. Then $d(v)=(6-\epsilon-\delta)\omega$ for some
$\delta\ge 0$. By the discharging rule, $v$ takes
$\tfrac{\delta}{(6-\epsilon-\delta)\omega}$ from the charge of each of
its $(6-\epsilon-\delta)\omega$ neighbors, and might also give some of its charge to its neighbors. Thus $w'(v)\le
6-\epsilon-\delta+(6-\epsilon-\delta)\omega\tfrac{\delta}{(6-\epsilon-\delta)\omega}=
6-\epsilon$.

Consider now a vertex $v$ of degree more than $(6-\epsilon)\omega$. Then $d(v)=(6-\epsilon+\delta)\omega$ for some
$0< \delta\le \epsilon$. By \cref{cor:radius12}, $d_{1/2}(v)\le
2(\epsilon-\delta)\omega$, and thus $D=D_1(v)\setminus D_{1/2}(v)$
contains at least
$(6-\epsilon+\delta)\omega-2(\epsilon-\delta)\omega\ge
(6-3\epsilon+3\delta)\omega$ vertices.

By \cref{lem:ad2v}, each vertex of $D$ has degree at most
$(6-\epsilon-\delta)\omega$. Observe that the function
$x\mapsto \tfrac{x}{6-\epsilon-x}$ is increasing for our choice of
$\epsilon$ and for $x\in [0,\epsilon]$, thus each
vertex of $D$ takes at least $\tfrac{\delta}{(6-\epsilon-\delta)\omega}$
from the charge of $v$ (while $v$ does not receive any charge from its neighbors). It follows that $$w'(v)\le  6-\epsilon+\delta-
(6-3\epsilon+3\delta)\omega
\tfrac{\delta}{(6-\epsilon-\delta)\omega}=6-\epsilon+\tfrac{2\epsilon\delta-4\delta^2}{6-\epsilon-\delta}
\le 5.68,$$ for any $0\le\delta\le \epsilon$ (for our
choice of $\epsilon)$.
\end{proof}

In the next section,  we propose an approach to improve upon
\cref{thm:ad} (and get closer to \cref{conj}) using Fourier
analysis. %We only obtain partial results, but we believe that the ideas might be useful. 

 \section{Fourier analysis and  the average
   degree of unit-disk graphs}\label{sec:fourier}

 Recall from the introduction that the disk clique number of a unit disk graph is the largest size of a clique contained within a disk of radius 1/2 (note that this depends on the embedding of the graph in the plane). In this section, we will leverage Fourier-analytic techniques in order to investigate the ratio between the average degree of unit-disk graphs and their disk clique number. Since those tools are not standard in the graph theory literature, we first motivate their introduction informally.

 A unit disk graph $G$ is at its core simply a (multi)set of points in the plane and thus can be represented by a (discontinuous) function $f:\mathbb{R}^2 \to \mathbb{N}$. In this language, the degree of a vertex $v$ can be readily computed as $d(v):=\int_{D_1(v)}f(u)du$, where the integral denotes the counting measure, i.e., a sum; and thus the sum of the degrees will be $\int_{\mathbb{R}^2}f(x)d(x)dx$, yielding average degree

 \[\frac{\int_{\mathbb{R}^2}f(x)d(x)dx}{\int_{\mathbb{R}^2}f(x)dx}= \frac{\int_{\mathbb{R}^2}f(x)\int_{D_1(v)}f(u)dudx}{\int_{\mathbb{R}^2}f(x)dx}.\]

 On the other hand, the disk clique number can also be formulated nicely, as the number of vertices in a disk of radius $1/2$ centered at $x$ is simply $\int_{u\in D_{1/2}(x)}f(u)du$, and thus the disk clique number is $\max_{x \in \mathbb{R}^2} \int_{D_{1/2}(x)}f(u)du$.

 Both the expressions of average degree and disk clique number can be
 expressed in a more compact way using the language of
 \emph{convolutions}. The convolution product $f * g$ is defined as
 $(f*g)(x)=\int_{y \in \mathbb{R}^2}f(x)g(x-y)$. If we denote by
   $\chi_r:\mathbb{R}^2\to \{0,1\}$ the indicator function of the
 disk of radius $r>0$ centered in $\mathbf{0}=(0,0)$, we then have for the average degree the compact formulation $\int_{\mathbb{R}^2}f \cdot (f*\chi_1)/||f||_1$, while the disk clique number is $||f*\chi_{1/2}||_{\infty}$, where we have used the functional analysis norms $||f||_1=\int |f|$ and $||f||_\infty=\max f$. Controlling the ratio between these two quantities is the aim of Question~\ref{q:1} below.

 This functional formulation suggests a natural attempt at maximizing the  ratio between the disk clique number and the average degree: we could try to look for a nonnegative function $f:\mathbb{R}^2 \rightarrow \mathbb{R}^{\geq 0}$ (which we would then discretize) of fixed average $||f||_1$ and such that the unit disk average $f*\chi_1$ is constant, while $f$ and hopefully $f*\chi_{1/2}$ are not (and thus the maximum value of the latter would be greater than the average value of $f$). In the one-dimensional case, any nonconstant $1$-periodic function would do, but how to reason about this in two dimensions?

 The analogy with periodic function in the one-dimensional case suggests that Fourier analysis might be helpful. Furthermore, the language of convolution products gives us a second hint that these quantities would be simpler to investigate in the Fourier domain: indeed, a key property of Fourier transforms is that they turn convolution products into usual products: $\widehat{f*g}=\widehat{f} \cdot \widehat{g}$. The Fourier transforms of the functions $\chi_r$ involve special functions called Bessel functions of the first kind, denoted by $J_1$. In line with the one-dimensional case, where being $1$-periodic amounts to having Fourier coefficients which are nonzero only at integer values, the analogue of $1$-periodic functions that we are looking for will turn out to be functions \emph{whose Fourier coefficients are only nonzero when $J_1$ is zero}. Discretizing such a function will bring us back to the realm of unit disk graphs, yielding examples with interesting properties.

In the remainder of this section, we formalize this idea properly, leading to a proof of Theorem~\ref{thm:fourier}. In order to be rigorous, some parts require analytical tools, for example the framework of \emph{tempered distributions}, for which we refer to standard analysis textbooks, e.g., Rudin~\cite{Rud73}. But the graph-theoretically minded reader can safely disregard these analytical issues and read through the text using the analogies that we have just described, thinking of a function as merely a (continuous version) of a unit disk graph, and of an integral as a sum.

\subsection{The question}

A function that will be crucial in the remainder is
$\chi_r:\mathbb{R}^2\to \{0,1\}$, the indicator
function of the disk of radius $r>0$ centered in $\mathbf{0}=(0,0)$.

\[
  \chi_r(x)=
  \begin{cases}
      1 & \text{if}\ \|x\|\le r, \\
      0 & \text{otherwise.}
    \end{cases}
\]
We will be mostly interested in $\chi_{1/2}$ and $\chi_1$.

\medskip

Let $f:\mathbb{R}^2\to \mathbb{R}^{\ge 0}$ be a Lebesgue-integrable
function. As $f(x)=|f(x)|$ for any $x\in \mathbb{R}^2$, we have  \[\int_{\mathbb{R}^2}f(x)dx=\int_{\mathbb{R}^2}|f(x)|dx=\|f\|_1.
\]

We will be interested in integrating $f$ on disks of radius $r>0$
centered in points $x$. As explained above, this can be described easily with a
convolution product.

\[
\int_{y\in D_r(x)}f(y)dy= \int_{\|x-y\|\le r}f(y)dy=\int_{\mathbb{R}^2}f(y)\chi_r(y-x)dy=(f*\chi_r)(x),
\]
where $D_r(x)$ denotes the disk of radius $r$ centered in $x$, and
$f*\chi_r$ denotes the convolution product of $f$ and $\chi_r$.

\begin{qn}\label{q:1}
What is the minimum constant
$c>0$ such that for any Lebesgue-integrable
function $f:\mathbb{R}^2\to \mathbb{R}^{\ge 0}$, \[\mu_f:=\int_{\mathbb{R}^2}f(x)\cdot (f*\chi_1)(x)dx
  \le c \cdot \|f\|_1\cdot \|f*\chi_{1/2}\|_\infty
  ?\]
\end{qn}

Note that   $D_1(x)$ can be covered by  a constant number of disks of
radius $\tfrac12$, so there is a constant $c>0$ such that $\|f*
\chi_1\|_\infty\le c \cdot  \|f*
\chi_{1/2}\|_\infty$, and thus the question above is well defined.

\medskip

Take $R=[0,N]^2$ and let $\mathbf{1}_R$ denote the indicator function
of $R$. Note
that $\|\mathbf{1}_R\|_1=N^2$ and $\|\mathbf{1}_R*\chi_{1/2}\|_\infty=
\pi/4$ for $N\ge 1$. On the other
hand, for all $x\in [1,N-1]^2$, $(\mathbf{1}_R*\chi_1)(x)= \pi$ and thus
$\mu_f=(\pi-o(1))N^2=(4-o(1)) \|\mathbf{1}_R\|_1\cdot \|\mathbf{1}_R*\chi_{1/2}\|_\infty$ , as $N\to \infty$. This
shows that $c\ge 4$ in Question~\ref{q:1}. We will see in
\cref{sub:ex} a finer example that shows that $c\ge 4.0905$ in Question~\ref{q:1}.

\subsection{Application to unit-disk graphs}\label{sec:fourierudg}

%Before we try to answer Question \ref{q:1}, 

Let us explain the
connection to our original problem of bounding the average degree
of unit-disk graphs as a function of their clique number.

%\footnote{After having
%  written this section, I realized that it might have been cleaner to
%  simply define $f$ as a finite linear combination of Dirac delta
%  functions in $\mathbb{R}^2$, but maybe doing so causes some
%  difficulties for square-integrability?}

Take a unit-disk graph $G$ embedded in the plane, and consider the associated (finite) point multiset $P=V(G)$ in
$\mathbb{R}^2$. By translation, we can assume that $P\subseteq
[0,N]^2$, for some real $N>0$. Fix some integer $M\gg N$, set $\delta:=N/M$; and divide $[0,N]^2$
into $M^2$ squares $R_{ij}$ (for $i,j\in \{1,\ldots,M\}$) of size
$\delta \times \delta$ (and area $\delta^2)$. Define the following
function $f:\mathbb{R}^2\to \mathbb{R}^{\ge 0}$.
\[
  f(x)=
  \begin{cases}
    0& \text{if}\ x\not\in [0,N]^2, \\
   \frac1{\delta^2} |P\cap R_{ij}| & \text{if}\ x\in R_{ij}.
  \end{cases}
\]
In words, $f(x)$ is  the number of points of $P$ in the square
containing $x$, scaled by a factor $\delta^2$.

Note that for any $i,j\in \{1,\ldots,M\}$,
\[\int_{R_{ij}}f(x)dx=\delta^2\cdot \tfrac1{ \delta^2}
|P\cap R_{ij}|= |P\cap R_{ij}|.\] It follows that 
$\|f\|_1=|P|=|V(G)|$. More generally for any fixed disk $D\subseteq \mathbb{R}^2$
of positive radius, $\int_{D}f(x)dx\to |P\cap D| $ as
$\delta\to 0$ (or equivalently, as $M\to \infty$).

As any disk of radius $\tfrac12$ has diameter at
most 1, the points of $P$ lying in such a disk form a clique in $G$,
and thus $P$ intersects every such clique in at most $\omega(G)$
points. It follows that $\|f*\chi_{1/2}\|\le (1+o(1))\omega(G)$ (where the $o(1)$
term is with respect to 
$\delta\to 0$, or equivalently, as $M\to \infty$). 

For sufficiently small $\delta>0$, for any point $p\in P
$ and $\delta\times\delta$ square $R$ containing $p$, $\int_R f(x)
(f*\chi_1)(x)dx$ counts the number of points of $P$ (possibly with repetitions)
coinciding with $p$, times the number of points in the disk of radius 1
centered in $p$ (which is the same as the number of neighbors of $p$
in $G$). It follows that as $\delta\to 0$, $\mu_f\to
\sum_{v\in V(G)}d(v)$, the sum of the degrees of $G$. So for any constant $c>0$ answering
Question \ref{q:1}, we have \[\sum_{v\in V(G)}d(v)\le
c\cdot  |V(G)|\cdot (1+o(1))\cdot \omega(G),\] and thus the average
degree of $G$ is at most  $(c+o(1))\cdot \omega(G)$.

\medskip

What about the reverse direction? Given some function $f$ showing that
the answer to Question~\ref{q:1} is at least $c_0$, for some $c_0>0$,
we can translate this into a discrete distribution of points in the
plane as follows. We consider a sufficiently large grid of dimension
$N\times N$, and sufficiently small step $\delta>0$, and for some sufficiently large
real $k>0$ we place $\lceil k f(x)\rceil$ points of $P$ at each point
$x$ of the grid. If $f$ behaves well (for instance if $f$ is continuous and periodic), the ratio between the average
degree and the disk clique number (recall that this is the maximum number of points in a disk of radius $\tfrac12$) of the associated unit-disk graph $G$ will be close to 
$c_0$. However it might be possible that maximum cliques in the unit-disk graph $G$
do not come from disks of radius $\tfrac12$, but other shapes of
diameter 1 (see \cref{sub:ex}). In this case the average degree of $G$ is not necessarily  close to $c_0 \cdot \omega(G)$.

%This is the case in the example
%of Subsection~\ref{sub:ex}, which shows that $c\ge c_0\approx 4.0905$ in
%Question~\ref{q:1}, but does not lead to a counterexample of \cref{conj}. Before we present this example, we need to introduce some
%tools from Fourier analysis.

\subsection{Tools}

Given a Lebesgue-integrable function $f:\mathbb{R}^2\to \mathbb{R}$, the Fourier transform
of $f$ is given by
\[
\hat{f}(\xi)=\int_{\mathbb{R}^2} f(x)e^{-2\pi i x\cdot \xi}dx,
\]
where $x\cdot \xi$ denotes the dot product of $x$ and $\xi$.
Note that for any function $f:\mathbb{R}^2\to \mathbb{R}^{\ge 0}$,

\[
\hat{f}(\mathbf{0})=\int_{\mathbb{R}^2} f(x)dx =\int_{\mathbb{R}^2} |f(x)|dx=\|f\|_1.
\]

When $\hat{f}$ is also integrable, we have that the reverse equality
\[
f(x)=\int_{\mathbb{R}^2} \hat{f}(\xi)e^{2\pi i x\cdot \xi}d\xi
\]
holds almost everywhere (and it holds everywhere if $f$ is
continuous).

\medskip

It is well known (see for example~\cite[Example~17.3.1]{handbook}) that for any $\xi\in \mathbb{R}^2$,
\[\hat{\chi_1}(\xi)=\frac{J_1(2\pi    \|\xi\|)}{\|\xi\|},
\]
where $J_1$ denotes the Bessel function of the first kind. Note that
$J_1(x)\sim x/2$ when $x\to 0$, so $\hat{\chi_1}(\mathbf{0})=\pi$ is well
defined by continuity.
Observe that for $r>0$
$\chi_r(x)=\chi_1(x/r)$ for every $x\in \mathbb{R}^2$, and thus by the
time-scaling property of Fourier transforms, we have 
\[\hat{\chi_r}(\xi)=r^2 \hat{\chi_1}(r\xi)=\frac{r J_1(2\pi r
    \|\xi\|)}{\|\xi\|}\] for any $\xi\in \mathbb{R}^2$ (as above,
$\hat{\chi_r}(\mathbf{0})=\pi r^2$ is well
defined by continuity).

\medskip

An
important property of the convolution product is that it behaves well
under Fourier transforms. For any $\xi\in \mathbb{R}^2$ and $r>0$,
\[
\widehat{(f*\chi_r)}(\xi)=\hat{f}(\xi)\hat{\chi_r}(\xi)=\hat{f}(\xi) \frac{J_1(2\pi r
    \|\xi\|)}{\|\xi\|}.
\]

So we have \[
(f*\chi_{1/2})(x)= \int_{\mathbb{R}^2} \hat{f}(\xi) \frac{J_1(\pi 
    \|\xi\|)}{\|\xi\|} e^{2\pi i x\cdot \xi }d\xi\le \|f*\chi_{1/2}\|_\infty ,
\] for any $x\in \mathbb{R}^2$.

\medskip

%Note that since $f$ is a real function, $\hat{f}$ is Hermitian:
%$\hat{f}(-\xi)=\overline{\hat{f}(\xi)}$ for any $\xi \in
%\mathbb{R}^2$.

If we denote by $\mathbf{1}$
the constant function with $\mathbf{1}(x)=1$ for any $x\in \mathbb{R}^2$, then
$\hat{\mathbf{1}}(\xi)=\delta(\xi)$, where $\delta$ denotes the Dirac delta
function\footnote{Dirac delta functions can be formally defined using distributions or generalized functions, see for example Rudin~\cite[Chapter~6]{Rud73}.}. In this case the equality above tells us that for any $x\in \mathbb{R}^2$,
\[
(\mathbf{1}*\chi_{r})(x)= \int_{\mathbb{R}^2} \delta(\xi) \frac{rJ_1(2\pi r
    \|\xi\|)}{\|\xi\|} e^{2\pi i x\cdot \xi }d\xi= \hat{\chi_r}(\mathbf{0})=\pi r^2,
\]
which is just a complicated way to say that the area of a disk of radius
$r$ is $\pi r^2$.

\medskip

For $f$ and $g$ square-integrable, the Parseval formula (see, e.g.,~\cite[Section~7.9]{Rud73}) stipulates that
\[\int_{\mathbb{R}^2}\overline{f(x)}g(x)dx=\int_{\mathbb{R}^2}\overline{\hat{f}(\xi)} \hat{g}(\xi)d\xi.\]

Assuming that $f$ is square-integrable, then so is $f * \chi_r$ (this follows from the Minkowski integral inequality~\cite[Inequality 202]{Har52}), and thus we have
\[
\int_{\mathbb{R}^2} \overline{f(x)}\cdot (f*\chi_r)(x)dx=\int_{\mathbb{R}^2} \overline{\hat{f}(\xi)}\cdot \widehat{(f*\chi_r)}(\xi)d\xi
\]
for any $r>0$.
If moreover, $f$ is a real
function, it
follows that
\[
\mu_f=\int_{\mathbb{R}^2}f(x)\cdot (f*\chi_1)(x)dx =\int_{\mathbb{R}^2} \overline{\hat{f}(\xi)}\cdot \hat{f}(\xi) \frac{J_1(2\pi
    \|\xi\|)}{\|\xi\|}d\xi = \int_{\mathbb{R}^2} \|\hat{f}(\xi)\|^2\cdot \frac{J_1(2\pi
    \|\xi\|)}{\|\xi\|}d\xi
\]

On the other hand, we have
\[
\int_{\mathbb{R}^2} \|\hat{f}(\xi)\|^2\cdot \frac{J_1(\pi
    \|\xi\|)}{2\|\xi\|}d\xi=\int_{\mathbb{R}^2}f(x)\cdot
  (f*\chi_{1/2})(x)dx \le \|f*\chi_{1/2}\|_\infty
  \int_{\mathbb{R}^2}f(x)dx\le \|f*\chi_{1/2}\|_\infty\cdot \|f\|_1
\]
So if we had $J_1(2\xi)\le \tfrac{c}2\cdot  J_1(\xi)$ for any $\xi\ge 0$, and
for some (hopefully small) constant $c>0$, this would directly imply $\mu_f\le c \|f\|_1\cdot \|f*\chi_{1/2}\|_\infty$. However, the former does not hold, as the two functions oscillate
independently.

%What about replacing that by $J_1(2\xi)\le \tfrac{c}2\cdot  J_1(\xi)+b$ for any
%$\xi\ge 0$? This would give an additional term of the form $b\cdot \int_{\mathbb{R}^2} \tfrac{\|\hat{f}(\xi)\|^2}{\|\xi\|}d\xi$.

%By the  Plancherel theorem, 
%$\int_{\mathbb{R}^2} %\|\hat{f}(\xi)\|^2d\xi=\int_{\mathbb{R}^2}
%|f(x)|^2dx$, but for us (in the application above), %typically $\int_{\mathbb{R}^2}
%|f(x)|^2dx$ tends to $\infty$ when $\delta\to 0$, so we have
%absolutely no control on a term like that. It is not clear if we
%have more control on the term with a $\|\xi\|$ at the denominator.

\medskip

An important observation here (in connection with the next section), is that for any $z\in [0,1.18]$, $J_1(2\pi z)\le
2\cdot  J_1(\pi z)$, with equality only if $z=0$. This shows that in the disk $D_{1.18}(\mathbf{0})$ of radius $1.18$
centered in $\mathbf{0}$, \[
  \int_{D_{1.18}(\mathbf{0})} \|\hat{f}(\xi)\|^2\cdot \frac{J_1(2\pi
    \|\xi\|)}{\|\xi\|}d\xi\le 4\cdot 
\int_{D_{1.18}(\mathbf{0})} \|\hat{f}(\xi)\|^2\cdot \frac{J_1(\pi
    \|\xi\|)}{2\|\xi\|}d\xi,
\]
so if the support of $\hat{f}$ lies in $D_{1.18}(\mathbf{0})$, we have
$\mu_f\le 4 \cdot \|f*\chi_{1/2}\|_\infty\cdot \|f\|_1$, with equality if
and only if the support of $\hat{f}$ is $\{0\}$, which is equivalent
to say that $f$ is a constant function. So if we want to find examples
where $\mu_f> 4 \cdot \|f*\chi_{1/2}\|_\infty\cdot \|f\|_1$, we need
to make sure that the support of $\hat{f}$ intersects the complement
of $D_{1.18}(\mathbf{0})$.

\subsection{A finer example}\label{sub:ex}

Consider the function $f:\mathbb{R}^2\to \mathbb{R}^{\ge 0}$ defined
by $f((x,y)):=1+\sin(2Bx)$, where $B$ is the first positive zero of the
Bessel function $J_1$, that is $B$ is the smallest positive real such
that $J_1(B)=0$.  It is known that
$B\approx 3.8317$. We will use Fourier analysis on $f$, which is definitely not Lebesgue-integrable. Nevertheless, this can be justified using the framework of \emph{tempered distributions}, to which $f$ belongs since it is bounded. For the sake of readability, we do not enter these technical details and refer the reader to standard textbooks on distributions, e.g., Rudin~\cite[Chapters~6 and~7]{Rud73}.

\medskip

Take some real number $N=k\cdot \tfrac{\pi}{B}$, for some integer
$k\ge 2$ (note that $\tfrac{\pi}{B}$  is the period of $f$), and let $g:\mathbb{R}^2\to
\mathbb{R}$ be defined as $g(x)=f(x)\cdot \mathbf{1}_{[0,N]^2}(x)$.
Note that  $\|g\|_1=N^2$ (since $N$ is a multiple of the period of $g$), and
$\|g*\chi_{1/2}\|_\infty=\|f*\chi_{1/2}\|_\infty$ (as $g$ is periodic
and we have chosen $N$ large enough). As we have seen
above, the value of $\|f*\chi_{1/2}\|_\infty$ can be
computed using \[
(f*\chi_{1/2})(x)= \int_{\mathbb{R}^2} \hat{f}(\xi) \frac{J_1(\pi 
    \|\xi\|)}{2\|\xi\|} e^{2\pi i x\cdot \xi }d\xi ,
\]
In order to use this equality, we observe that for any $\xi\in \mathbb{R}^2$,
\[
  \hat{f}(\xi)=\delta(\xi)+\tfrac1{2i}(\delta(\xi-\tfrac{2B}{2\pi})-\delta(\xi+\tfrac{2B}{2\pi})),
  \] where
$\delta$ denotes the Dirac delta function. This shows that
$\hat{f}(\xi)=0$ unless $\xi=0$ or $\xi=\pm \tfrac{B}{\pi}$ (we
observe, in connection with the final comment of the previous section,
that $\tfrac{B}{\pi}\approx 1.22>1.18$). If $\xi=\pm
\tfrac{B}{\pi}$, then $J_1(\pi \|\xi\|)=0$, by definition of $B$. It follows
that the
only non-zero term in the integral is for $\xi=0$, where $\frac{J_1(\pi 
    \|\xi\|)}{2\|\xi\|} e^{2\pi i x\cdot \xi }=\tfrac{\pi}4$. By the
  definition of the Dirac delta function, the integral evaluates as \[
(f*\chi_{1/2})(x)= \int_{\mathbb{R}^2} \hat{f}(\xi) \frac{J_1(\pi 
    \|\xi\|)}{2\|\xi\|} e^{2\pi i x\cdot \xi }d\xi =\frac{\pi}4,
\]
This shows that
$\|g*\chi_{1/2}\|_\infty=\|f*\chi_{1/2}\|_\infty=\tfrac{\pi}4$.

It remains to evaluate \[\mu_g=\int_{\mathbb{R}^2}g(x)\cdot
  (g*\chi_1)(x)dx=\int_{[0,N]^2}f(x)\cdot
  (g*\chi_1)(x)dx.\]
We can compute $\mu_g$ for $N$ a multiple of $\pi/B$ by using the Parseval formula as before, this time to $f$ and $g * \chi_1$. Note that $f$ is not square-integrable in this case, but it is a tempered distribution and $g$ has compact support and is thus a test function, so we can still use Parseval formula~\cite[Chapter~7]{Rud73} in this case (we omit the definitions of tempered distributions and test functions here, the reader is referred to \cite{Rud73} for more details). We first compute that $\hat{g}(0)=N^2=\|g\|_1$,  $\hat{g}(\frac{2B}{2\pi})=\frac{N^2}{2i}$ and $\hat{g}(\frac{-2B}{2\pi})=\frac{-N^2}{2i}$. Now:
\begin{eqnarray*}
\mu_g &=& \int_{[0,N]^2}f(x)\cdot (g*\chi_1)(x)dx \\& =& \int_{\mathbb{R}^2}f(x)\cdot (g*\chi_1)(x)dx - \int_{[-1,N+1]^2 \setminus [0,N]^2}f(x)\cdot (g*\chi_1)(x)dx\\
&= &  \int_{\mathbb{R}^2}f(x)\cdot (g*\chi_1)(x)dx - O(N) \\
&=&\int_{\mathbb{R}^2} \overline{\hat{f}(\xi)} \cdot \hat{g}(\xi)\cdot \frac{J_1(2\pi
    \|\xi\|)}{\|\xi\|}d\xi - O(N)\\
     &=& \int_{\mathbb{R}^2}  \left(\delta(\xi)+\tfrac1{2i}(\delta(\xi+\tfrac{2B}{2\pi})-\delta(\xi-\tfrac{2B}{2\pi}))\right)\cdot \hat{g}(\xi)\cdot \frac{J_1(2\pi
    \|\xi\|)}{\|\xi\|}d\xi  -O(N)\\
     &=& \int_{\mathbb{R}^2} \|g\|_1 \left(\lim_{\xi \rightarrow 0} \left(\frac{J_1(2\pi\|\xi\|)}{\|\xi\|}\right)+\frac{1}{4}\cdot\frac{\pi J_1(2B)}{B}+\frac{1}{4}\cdot\frac{\pi J_1(-2B)}{B}\right)-O(N)\\
       &=& \pi\|g\|_1\left(1+\frac{J_1(2B)}{2B}\right)-O(N),
\end{eqnarray*}
where we used that $f(x) \cdot (g * \chi_1(x))$ is zero outside of $[-1,N+1]^2$ and is always at most $2$, that $\lim_{\xi \rightarrow 0} \frac{J_1(2\pi\|\xi\|)}{\|\xi\|}=\pi$ and that the $J_1$ function is even. Now, $\frac{J_1(2B)}{2B}\approx 0.04527$ and $\|g\|_1 = N^2$, and thus we obtain that for $N$ large enough,
\[\mu_g \approx 4.0905 \cdot \|g\|_1 \cdot \|g* \chi_{1/2}\|_\infty.\]

Therefore we have $c\ge 4.0905$ in Question
\ref{q:1}. Recall that the disk clique number of a unit-disk graph $G$, denoted by $\omega_D(G)$, is the largest size of a set of points contained within a disk of radius $\tfrac12$ in $G$. Discretizing the function $g$ introduced above as described in \cref{sec:fourierudg} yields the following theorem.

%\begin{thm}\label{thm:fourier2}
%There exists a unit-disk graph $G$ of average degree at least $4.0905\, \omega_D(G)$.
%\end{thm}

\fourier*
%\begin{thm}\label{thm:fourier2}
%There exists a unit-disk graph $G$ of average degree at least $4.0905\, \omega_D(G)$.
%\end{thm}

\begin{proof}
We consider a 2-dimensional grid of dimension $N \times N$ and step $\delta >0$, and we place $\lceil k g(x)\rceil$ points at each point $x$ of the grid for the function $g(x):=(1+\sin(2Bx))\cdot \mathbf{1}_{[0,N]^2}(x)$. This set of points defines a unit-disk graph $G$. As explained in \cref{sec:fourierudg}, for $N$ and $k$ large enough and $\delta$ small enough, the ratio between the average degree of $G$ and its disk clique number $\omega_D(G)$ will converge to $\frac{\|g \cdot (g * \chi_1)\|_1}{\|g\|_1\|g *\chi_{1/2}\|_\infty}$. By the calculations above, this ratio is larger than $4.0905$ for sufficiently large $N$, thus establishing \cref{thm:fourier}. 
\end{proof}

Note however that (as alluded to in  \cref{sec:fourierudg}), the unit-disk graph $G$ obtained from this construction does not  contradict \cref{conj}. Indeed,  numerical computations\footnote{As the code is not easily readable, due to several (probably necessary) optimization tricks, we chose not to make it publicly available, however we are happy to send it to any interested reader upon request. The computation boils down to solving a max-flow instance in a fairly dense graph on $\sim 10^{4}$ vertices.} suggest
that some smooth versions of Reuleaux triangles are slightly denser than disks of radius
$\tfrac12$ with respect to $f$ (and $g$), and in fact these
computations seem to indicate that the ratio between the average
degree and the clique number in the corresponding unit-disk graph $G$ is close to 3.93 (see
\cref{fig:clique} for a picture of a maximum clique with respect to
this distribution, with density significantly larger than $\pi/4$).

\begin{figure}[htb]
 \centering
 \includegraphics[scale=0.3]{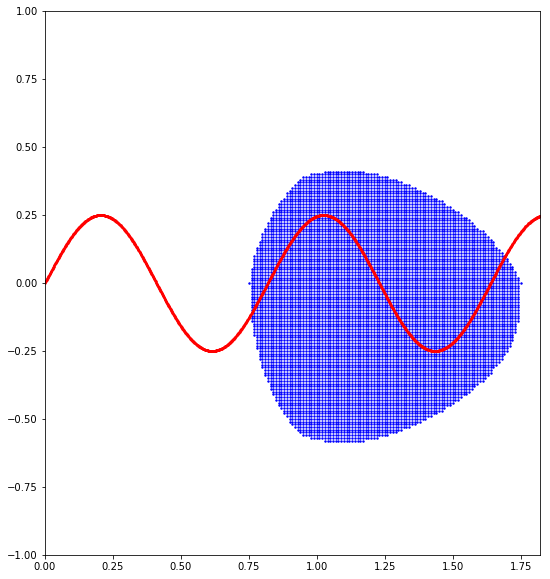}
 \caption{The points of a maximum clique with respect to $f:(x,y)\mapsto
   1+\sin(2Bx)$. The function $f$ is plotted in red, but not on the same
   scale as the points of the clique (in blue).}
 \label{fig:clique}
\end{figure}

\section{Conclusion}

Given a sequence of pairs $\mathcal{P}=(p_i,r_i)_{1\le i \le n}$,
where each $p_i$ is a point in the plane and each $r_i$ is a positive
real, the
\emph{disk graph} on $\mathcal{P}$ is the graph with vertex set
$\{p_1,p_2,\ldots,p_n\}$, in which two vertices $p_i$ and $p_j$ are
adjacent if and only if $\|p_i-p_j\|\le r_i+r_j$. Disk graphs model wireless
communication networks using omni-directional antennas of possibly
different powers (where the power of the $i$-th antenna is proportional to the real number $r_i$), and disk graphs where all the reals $r_i$ are equal
are precisely unit-disk graphs. A natural question is whether results
similar to the results we obtain here can be proved for disk graphs. A
major difference is that disk graphs do not have their maximum degree
bounded by a function of their clique number $\omega$ (the class of disk
graphs contains the class of all trees for instance). However disk
graphs have average degree bounded by a constant times $\omega$, and
this can be used to obtain $O(\log n)$ round algorithms coloring these
graphs with few colors~\cite{BE10}. In the same spirit as \cref{q:mincollog} we
can ask the following.

\begin{qn}
What is the minimum real $c>0$ such that a coloring of any $n$-vertex
disk graph $G$ with $c \cdot \omega(G)$ colors can be obtained in
$O(\log n)$ rounds in the \textsf{LOCAL} model?
\end{qn}

%A final comment is that all the results and techniques developed in
%this paper also hold for unit-disk graphs embedded on the torus
%(rather than the plane), provided the area of the torus is
%sufficiently large (see~\cite{CAJ07} for more on unit-disk graphs on
%the torus). It does not seem to be known that for these graphs
%$\chi(G)\le 3 \omega(G)$. The best that can be achieved using existing
%techniques seems to be  $\chi(G)\le (3+\epsilon) \omega(G)$, where $\epsilon
%\to 0$ as the area of the torus tends to $\infty$ (regardless of
%complexity concerns). It follows that with our current approach, avoiding the
%$\epsilon$ in the bound \cref{thm:geom} might be a difficult task.

As a final comment, we recall that in the proceedings version of the paper~\cite{proc} we  claimed that for any $\epsilon>0$, any unit-disk graph $G$ could be colored with $(3+\epsilon)\omega(G)$ colors in $O(1/\epsilon)$ rounds in the \textsf{location-aware LOCAL} model, but that we later found an error in our original argument, and could only replace $(3+\epsilon)\omega(G)$ by $4\omega(G)$ (in Theorem \ref{thm:geom}). A problem left open by this work is to prove (or disprove) that $(3+\epsilon)\omega(G)$ colors are sufficient. A related (purely existential) question is the following. A \emph{flat annulus}, or \emph{cylinder}, of height $h$ and circumference $\ell$ is the metric space obtained from a Euclidean rectangle of height $h$ and length $\ell$ by identifying its left and right boundaries. While a rectangle is the Cartesian product of two segments, a flat annulus can be thought of as the product of a cycle with a segment. 

\begin{qn}
Let $\epsilon>0$. Let $G$ be a unit-disk graph $G$ whose vertices are embedded in a flat annulus of height $\sqrt{3}/2$ and sufficiently large circumference (as a function of $\epsilon$). Is it true that if $\omega(G)$ is sufficiently large, $\chi(G)\le (1+\epsilon)\omega(G)$?
\end{qn}

%\bigskip

\begin{acknowledgement} The authors would like to thank Wouter Cames
  van Batenburg and Fran\c{c}ois Pirot for the interesting discussions. The authors would also like to express their gratitude to the reviewers of the conference and journal versions of the paper for their helpful comments and suggestions, and to Mohsen Ghaffari for his kind explanations on the status of the $(\text{deg}+\epsilon \Delta)$-list coloring and $(\text{deg}+1)$-list coloring problems.
\end{acknowledgement}

     \appendix
 
 \section{Computations}\label{sec:comput}

Computations were done with SageMath 9.2, which is a free open-source mathematics
software, downloadable from \url{https://www.sagemath.org/}.

  \begin{tcolorbox}[breakable, size=fbox, boxrule=1pt, pad at break*=1mm,colback=cellbackground, colframe=cellborder]
\prompt{In}{incolor}{1}{\boxspacing}
\begin{Verbatim}[commandchars=\\\{\}]
\PY{k}{def} \PY{n+nf}{solLP5}\PY{p}{(}\PY{n}{d}\PY{p}{)}\PY{p}{:}
    \PY{n}{lp5} \PY{o}{=} \PY{n}{MixedIntegerLinearProgram}\PY{p}{(}\PY{p}{)}
    \PY{n}{v} \PY{o}{=} \PY{n}{lp5}\PY{o}{.}\PY{n}{new\PYZus{}variable}\PY{p}{(}\PY{n}{real}\PY{o}{=}\PY{k+kc}{True}\PY{p}{,} \PY{n}{nonnegative}\PY{o}{=}\PY{k+kc}{True}\PY{p}{)}
    \PY{n}{x1}\PY{p}{,} \PY{n}{x2}\PY{p}{,} \PY{n}{x3}  \PY{o}{=} \PY{n}{v}\PY{p}{[}\PY{l+s+s1}{\PYZsq{}}\PY{l+s+s1}{x1}\PY{l+s+s1}{\PYZsq{}}\PY{p}{]}\PY{p}{,} \PY{n}{v}\PY{p}{[}\PY{l+s+s1}{\PYZsq{}}\PY{l+s+s1}{x2}\PY{l+s+s1}{\PYZsq{}}\PY{p}{]}\PY{p}{,} \PY{n}{v}\PY{p}{[}\PY{l+s+s1}{\PYZsq{}}\PY{l+s+s1}{x3}\PY{l+s+s1}{\PYZsq{}}\PY{p}{]}
    \PY{n}{lp5}\PY{o}{.}\PY{n}{set\PYZus{}objective}\PY{p}{(}\PY{n}{x1}\PY{o}{+}\PY{n}{x2}\PY{o}{+}\PY{l+m+mi}{2}\PY{o}{*}\PY{n}{x3}\PY{p}{)}
    \PY{n}{lp5}\PY{o}{.}\PY{n}{add\PYZus{}constraint}\PY{p}{(}\PY{n}{x2}\PY{o}{+}\PY{n}{x3}\PY{o}{\PYZlt{}}\PY{o}{=}\PY{l+m+mi}{5}\PY{p}{)}
    \PY{n}{lp5}\PY{o}{.}\PY{n}{add\PYZus{}constraint}\PY{p}{(}\PY{p}{(}\PY{l+m+mi}{2}\PY{o}{\PYZhy{}}\PY{n}{sqrt}\PY{p}{(}\PY{n}{d}\PY{o}{\PYZca{}}\PY{l+m+mi}{2}\PY{o}{/}\PY{l+m+mi}{4}\PY{o}{+}\PY{p}{(}\PY{l+m+mi}{1}\PY{o}{\PYZhy{}}\PY{n}{sqrt}\PY{p}{(}\PY{l+m+mi}{1}\PY{o}{\PYZhy{}}\PY{n}{d}\PY{o}{\PYZca{}}\PY{l+m+mi}{2}\PY{o}{/}\PY{l+m+mi}{4}\PY{p}{)}\PY{p}{)}\PY{o}{\PYZca{}}\PY{l+m+mi}{2}\PY{p}{)}\PY{p}{)}\PY{o}{*}\PY{n}{x1}\PY{o}{+}\PY{n}{x2}\PY{o}{+}\PY{p}{(}\PY{l+m+mi}{3}\PY{o}{\PYZhy{}}\PY{n}{sqrt}\PY{p}{(}\PY{l+m+mi}{1}\PY{o}{+}\PY{n}{d}\PY{o}{\PYZca{}}\PY{l+m+mi}{2}\PY{o}{\PYZhy{}}\PY{n}{d}\PY{o}{*}\PY{n}{sqrt}\PY{p}{(}\PY{l+m+mi}{3}\PY{p}{)}\PY{p}{)}\PY{p}{)}\PY{o}{*}\PY{n}{x3}\PY{o}{\PYZlt{}}\PY{o}{=}\PY{l+m+mi}{6}\PY{p}{)}
    \PY{k}{return} \PY{n+nb}{round}\PY{p}{(}\PY{n}{lp5}\PY{o}{.}\PY{n}{solve}\PY{p}{(}\PY{p}{)}\PY{p}{,} \PY{l+m+mi}{4}\PY{p}{)}
\end{Verbatim}
\end{tcolorbox}

    \begin{tcolorbox}[breakable, size=fbox, boxrule=1pt, pad at break*=1mm,colback=cellbackground, colframe=cellborder]
\prompt{In}{incolor}{2}{\boxspacing}
\begin{Verbatim}[commandchars=\\\{\}]
\PY{k}{def} \PY{n+nf}{solLP4}\PY{p}{(}\PY{n}{d}\PY{p}{)}\PY{p}{:}
    \PY{n}{lp4} \PY{o}{=} \PY{n}{MixedIntegerLinearProgram}\PY{p}{(}\PY{p}{)}
    \PY{n}{v} \PY{o}{=} \PY{n}{lp4}\PY{o}{.}\PY{n}{new\PYZus{}variable}\PY{p}{(}\PY{n}{real}\PY{o}{=}\PY{k+kc}{True}\PY{p}{,} \PY{n}{nonnegative}\PY{o}{=}\PY{k+kc}{True}\PY{p}{)}
    \PY{n}{x1}\PY{p}{,} \PY{n}{x2}\PY{p}{,} \PY{n}{x3}\PY{p}{,} \PY{n}{x4}  \PY{o}{=} \PY{n}{v}\PY{p}{[}\PY{l+s+s1}{\PYZsq{}}\PY{l+s+s1}{x1}\PY{l+s+s1}{\PYZsq{}}\PY{p}{]}\PY{p}{,} \PY{n}{v}\PY{p}{[}\PY{l+s+s1}{\PYZsq{}}\PY{l+s+s1}{x2}\PY{l+s+s1}{\PYZsq{}}\PY{p}{]}\PY{p}{,} \PY{n}{v}\PY{p}{[}\PY{l+s+s1}{\PYZsq{}}\PY{l+s+s1}{x3}\PY{l+s+s1}{\PYZsq{}}\PY{p}{]}\PY{p}{,} \PY{n}{v}\PY{p}{[}\PY{l+s+s1}{\PYZsq{}}\PY{l+s+s1}{x4}\PY{l+s+s1}{\PYZsq{}}\PY{p}{]}
    \PY{n}{lp4}\PY{o}{.}\PY{n}{set\PYZus{}objective}\PY{p}{(}\PY{n}{x1}\PY{o}{+}\PY{n}{x2}\PY{o}{+}\PY{l+m+mi}{2}\PY{o}{*}\PY{n}{x3}\PY{o}{+}\PY{l+m+mi}{2}\PY{o}{*}\PY{n}{x4}\PY{p}{)}
    \PY{n}{lp4}\PY{o}{.}\PY{n}{add\PYZus{}constraint}\PY{p}{(}\PY{n}{x2}\PY{o}{+}\PY{n}{x3}\PY{o}{+}\PY{n}{x4}\PY{o}{\PYZlt{}}\PY{o}{=}\PY{l+m+mi}{4}\PY{p}{)}
    \PY{n}{lp4}\PY{o}{.}\PY{n}{add\PYZus{}constraint}\PY{p}{(}\PY{n}{x1}\PY{o}{+}\PY{l+m+mi}{2}\PY{o}{*}\PY{n}{x3}\PY{o}{\PYZlt{}}\PY{o}{=}\PY{l+m+mi}{2}\PY{p}{)}
    \PY{n}{lp4}\PY{o}{.}\PY{n}{add\PYZus{}constraint}\PY{p}{(}\PY{p}{(}\PY{l+m+mi}{2}\PY{o}{\PYZhy{}}\PY{n}{sqrt}\PY{p}{(}\PY{n}{d}\PY{o}{\PYZca{}}\PY{l+m+mi}{2}\PY{o}{/}\PY{l+m+mi}{4}\PY{o}{+}\PY{p}{(}\PY{n}{sqrt}\PY{p}{(}\PY{l+m+mi}{3}\PY{p}{)}\PY{o}{/}\PY{l+m+mi}{2}\PY{o}{\PYZhy{}}\PY{n}{sqrt}\PY{p}{(}\PY{l+m+mi}{1}\PY{o}{\PYZhy{}}\PY{p}{(}\PY{l+m+mi}{1}\PY{o}{+}\PY{n}{d}\PY{p}{)}\PY{o}{\PYZca{}}\PY{l+m+mi}{2}\PY{o}{/}\PY{l+m+mi}{4}\PY{p}{)}\PY{p}{)}\PY{o}{\PYZca{}}\PY{l+m+mi}{2}\PY{p}{)}\PY{p}{)}\PY{o}{*}\PY{n}{x1}\PY{o}{+}\PY{n}{x2}\PY{o}{+}\PY{p}{(}\PY{l+m+mi}{3}\PY{o}{\PYZhy{}}\PY{n}{sqrt}\PY{p}{(}\PY{l+m+mi}{1}\PY{o}{+}\PY{n}{d}\PY{o}{\PYZca{}}\PY{l+m+mi}{2}\PY{o}{\PYZhy{}}\PY{n}{d}\PY{p}{)}\PY{p}{)}\PY{o}{*}\PY{n}{x3}\PY{o}{+}\PY{p}{(}\PY{l+m+mi}{3}\PY{o}{\PYZhy{}}\PY{n}{sqrt}\PY{p}{(}\PY{l+m+mi}{1}\PY{o}{\PYZhy{}}\PY{n}{d}\PY{p}{)}\PY{p}{)}\PY{o}{*}\PY{n}{x4}\PY{o}{\PYZlt{}}\PY{o}{=}\PY{l+m+mi}{6}\PY{p}{)}
    \PY{k}{return} \PY{n+nb}{round}\PY{p}{(}\PY{n}{lp4}\PY{o}{.}\PY{n}{solve}\PY{p}{(}\PY{p}{)}\PY{p}{,} \PY{l+m+mi}{4}\PY{p}{)}
\end{Verbatim}
\end{tcolorbox}

    \begin{tcolorbox}[breakable, size=fbox, boxrule=1pt, pad at break*=1mm,colback=cellbackground, colframe=cellborder]
\prompt{In}{incolor}{3}{\boxspacing}
\begin{Verbatim}[commandchars=\\\{\}]
\PY{n}{p} \PY{o}{=} \PY{n}{list\PYZus{}plot}\PY{p}{(}\PY{p}{[}\PY{n}{solLP5}\PY{p}{(}\PY{n}{x}\PY{p}{)} \PY{k}{for} \PY{n}{x} \PY{o+ow}{in} \PY{n}{srange}\PY{p}{(}\PY{l+m+mf}{0.5}\PY{p}{,} \PY{l+m+mi}{1}\PY{p}{,} \PY{l+m+mf}{0.005}\PY{p}{,} \PY{n}{include\PYZus{}endpoint}\PY{o}{=}\PY{k+kc}{True}\PY{p}{)}\PY{p}{]}\PY{p}{,} \PY{n}{plotjoined}\PY{o}{=}\PY{k+kc}{True}\PY{p}{)}
\PY{n}{show}\PY{p}{(}\PY{n}{p}\PY{p}{)}
\end{Verbatim}
\end{tcolorbox}

    \begin{center}
    \adjustimage{max size={0.3\linewidth}{0.3\paperheight}}{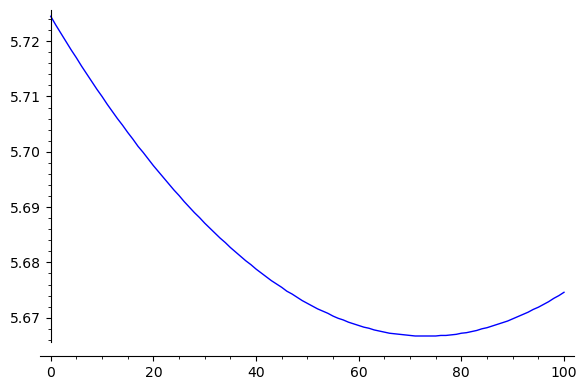}
    \end{center}
 %   { \hspace*{\fill} \\}
    
    \begin{tcolorbox}[breakable, size=fbox, boxrule=1pt, pad at break*=1mm,colback=cellbackground, colframe=cellborder]
\prompt{In}{incolor}{4}{\boxspacing}
\begin{Verbatim}[commandchars=\\\{\}]
\PY{n}{p} \PY{o}{=} \PY{n}{list\PYZus{}plot}\PY{p}{(}\PY{p}{[}\PY{n}{solLP4}\PY{p}{(}\PY{n}{x}\PY{p}{)} \PY{k}{for} \PY{n}{x} \PY{o+ow}{in} \PY{n}{srange}\PY{p}{(}\PY{l+m+mf}{0.5}\PY{p}{,} \PY{l+m+mi}{1}\PY{p}{,} \PY{l+m+mf}{0.005}\PY{p}{,} \PY{n}{include\PYZus{}endpoint}\PY{o}{=}\PY{k+kc}{True}\PY{p}{)}\PY{p}{]}\PY{p}{,} \PY{n}{plotjoined}\PY{o}{=}\PY{k+kc}{True}\PY{p}{)}
\PY{n}{show}\PY{p}{(}\PY{n}{p}\PY{p}{)}
\end{Verbatim}
\end{tcolorbox}

    \begin{center}
    \adjustimage{max size={0.3\linewidth}{0.3\paperheight}}{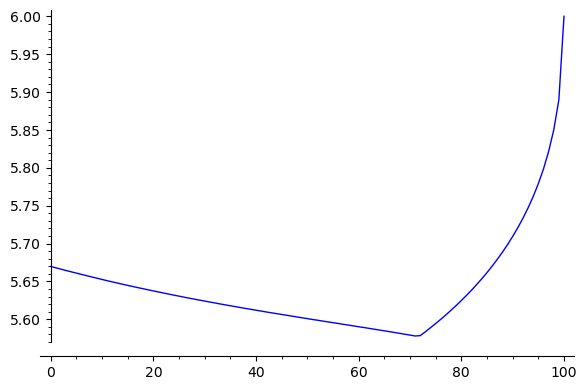}
    \end{center}
  %  { \hspace*{\fill} \\}
    
    \begin{tcolorbox}[breakable, size=fbox, boxrule=1pt, pad at break*=1mm,colback=cellbackground, colframe=cellborder]
\prompt{In}{incolor}{5}{\boxspacing}
\begin{Verbatim}[commandchars=\\\{\}]
\PY{n}{solLP5}\PY{p}{(}\PY{l+m+mi}{1}\PY{p}{)}
\end{Verbatim}
\end{tcolorbox}

            \begin{tcolorbox}[breakable, size=fbox, boxrule=.5pt, pad at break*=1mm, opacityfill=0]
\prompt{Out}{outcolor}{6}{\boxspacing}
\begin{Verbatim}[commandchars=\\\{\}]
5.6746
\end{Verbatim}
\end{tcolorbox}
        
    \begin{tcolorbox}[breakable, size=fbox, boxrule=1pt, pad at break*=1mm,colback=cellbackground, colframe=cellborder]
\prompt{In}{incolor}{7}{\boxspacing}
\begin{Verbatim}[commandchars=\\\{\}]
\PY{n}{solLP4}\PY{p}{(}\PY{l+m+mi}{1}\PY{o}{/}\PY{l+m+mi}{2}\PY{p}{)}
\end{Verbatim}
\end{tcolorbox}

            \begin{tcolorbox}[breakable, size=fbox, boxrule=.5pt, pad at break*=1mm, opacityfill=0]
\prompt{Out}{outcolor}{8}{\boxspacing}
\begin{Verbatim}[commandchars=\\\{\}]
5.6698
\end{Verbatim}
\end{tcolorbox}

 \end{document}